\documentclass{ijuc} 
\usepackage{fullpage}
\usepackage{epsfig}
\usepackage{amssymb}
\usepackage{amsmath}
\usepackage{latexsym} 
\begin{document} 

\newtheorem{theorem}{Theorem} 
\newtheorem{corollary}[theorem]{Corollary} 
\newtheorem{prop}[theorem]{Proposition} 
\newtheorem{problem}{Problem} 
\newtheorem{lemma}[theorem]{Lemma} 
\newtheorem{observation}[theorem]{Observation} 
\newtheorem{defin}[theorem]{Definition} 
\newtheorem{example}[theorem]{Example} 
\newenvironment{proof}{{\bf Proof:\ }}{\hfill $\Box$\linebreak\vskip2mm}
\newcommand{\PR}{{\bf Proof:\ }}             
\def\EPR{\hfill $\Box$\linebreak\vskip2mm}  
 
\def\Pol{{\sf Pol}} 
\def\Polo{{\sf Pol}_1\;}
\def\PPol{{\sf pPol}} 
\def\Inv{{\sf Inv}} 
\def\Clo{{\sf Clo}\;} 
\def\cP{\mathcal{P}}
\def\cC{\mathcal{C}}
\def\cF{\mathcal{F}}
\def\Alg{{\sf Alg}}
\def\GF{\mathrm{GF}}
 
\let\op=\oplus  
\let\omn=\ominus  
\let\meet=\wedge  
\let\join=\vee  
\let\tm=\times  
  
\let\wh=\widehat  
\def\ox{\ov x}  
\def\oy{\ov y}  
\def\oz{\ov z}  
\def\of{\ov f}  
\def\oa{\ov a}  
\def\ob{\ov b}  
\def\oc{\ov c} 
\let\un=\underline  
\let\ov=\overline  
  
\def\ba{{\bf a}}  
\def\bb{{\bf b}}  
\def\bc{{\bf c}}  
\def\bd{{\bf d}}  
\def\be{{\bf e}}  
\def\bs{{\bf s}}  
\def\bx{{\bf x}}  
\def\by{{\bf y}}  
\def\bw{{\bf w}}  
\def\bz{{\bf z}}  
\def\CSP{{\rm CSP}} 
\def\WCSP{{\rm WCSP}} 
\def\NCSP{{\rm \#CSP}}
\def\FP{{\rm FP}} 
\def\CNF{{\rm CNF}} 
\def\NAND{{\sf NAND}}
\def\orr{{\sf or}}
\def\horn{{\sf horn}}
\def\Compl{{\sf Compl}}
\def\OR{{\sf OR}}
\def\Cl{{\sf Cl}}
\def\NEQ{{\sf NEQ}}
\def\IMP{{\sf IMP}}
\def\NIMP{{\sf NIMP}}
\def\EQ{{\sf EQ}}
\def\NAE{{\sf NAE}}
 
\let\sse=\subseteq  
\def\ang#1{\langle #1 \rangle}  
\def\dang#1{\ang{\ang{#1}}}  
\def\mang#1{\langle #1 \rangle_{\max}}  
\def\mange#1{\langle #1 \rangle^1_{\max}}  
\def\vc#1#2{#1 _1\zd #1 _{#2}}  
\def\zd{,\ldots,}  
\def\cl#1#2{  
\left(\begin{array}{c} #1\\ #2 \end{array}\right)}  
\def\cll#1#2#3{  
\left(\begin{array}{c} #1\\ #2\\ #3 \end{array}\right)}  
\let\ex=\exists
\def\rel{R}
\def\relo{Q}
\def\rela{S}
\def\lb{$\linebreak$}  
\def\ar{{\sf ar}}
\def\pr{{\rm pr}}
\def\Dom{{\sf Dom}}
  
\let\al=\alpha  
\let\gm=\gamma  
\let\dl=\delta
\let\ve=\varepsilon  
\let\ld=\lambda  
\let\vf=\varphi  
\let\Gm=\Gamma  
\let\Dl=\Delta  

\font\tengoth=eufm10 scaled 1200  
\font\sixgoth=eufm6  
\def\goth{\fam12}  
\textfont12=\tengoth \scriptfont12=\sixgoth \scriptscriptfont12=\sixgoth  
\font\tenbur=msbm10  
\font\eightbur=msbm8  
\def\bur{\fam13}  
\textfont11=\tenbur \scriptfont11=\eightbur \scriptscriptfont11=\eightbur  
\font\twelvebur=msbm10 scaled 1200  
\textfont13=\twelvebur \scriptfont13=\tenbur  
\scriptscriptfont13=\eightbur  
\mathchardef\nat="0B4E  
\mathchardef\eps="0D3F  
  
\def\mkpPol{{\rm m}(k){\sf -pPol}}
\def\mkPol{{\rm m}(k){\sf -Pol}}
\def\mKpPol{{\rm m}(K){\sf -pPol}}
\def\mKPol{{\rm m}(K){\sf -Pol}}
\def\mpPol{{\rm m}{\sf -pPol}}
\def\mPol{{\rm m}{\sf -Pol}}
\def\kex{\exists_k}
\def\mex{\exists_{\max}}
\def\mexe{\exists^1_{\max}}
\def\mkP#1{P^{k,(#1)}_D}
\def\mkF#1{F^{k,(#1)}_D}
\def\mP#1{P^{(#1)}_D}
\def\mF#1{F^{(#1)}_D}
\def\mkPa{P^k_D}
\def\mkFa{F^k_D}
\def\mPa{P_D}
\def\mFa{F_D}

\title{Galois correspondence for counting quantifiers}
\author{Andrei A. Bulatov\inst{1}\email{abulatov@cs.sfu.ca} 
\and Amir Hedayaty\inst{1}\email{aha49@cs.sfu.ca}
}

\institute{Simon Fraser University}
\date{}
\def\received{}
\maketitle

\begin{abstract}
We introduce a new type of closure operator on the set of relations, 
max-implementation, and its weaker analog max-quantification. Then 
we show that approximation preserving reductions between counting constraint 
satisfaction problems (\#CSPs) are preserved by these two types of 
closure operators. Together with some previous results this means that 
the approximation complexity of counting CSPs is determined by partial 
clones of relations that additionally closed under these new types of 
closure operators. Galois correspondence of various kind have proved to 
be quite helpful in the study of the complexity of the CSP. While we 
were unable to identify a Galois correspondence for partial clones closed 
under max-implementation and max-quantification, we obtain such results 
for slightly different type of closure operators, $k$-existential quantification. 
This type of quantifiers are known as counting quantifiers in model theory, 
and often used to enhance first order logic languages. We characterize 
partial clones of relations closed under $k$-existential quantification as 
sets of relations invariant under a set of partial functions that satisfy the 
condition of $k$-subset surjectivity. Finally, we give a description of 
Boolean max-co-clones, that is, sets of relations on $\{0,1\}$ closed under
max-implementations.

This is an extended version of \cite{Bulatov12:quantifiers}.

\end{abstract}

\section{Introduction}

Clones of functions and clones of relations in their various incarnations 
have proved to be an immensely powerful tool in the study of the 
complexity of different versions of the Constraint Satisfaction Problem 
(CSP, for short). In a CSP the aim is to find an assignment of values to 
a given set of variables, subject to constraints on the values that can be 
assigned simultaneously to certain specified subsets of variables. A CSP 
can also be expressed as the problem of deciding whether a given 
conjunctive formula has a model. In the counting version of the CSP the 
goal is to find the number of satisfying assignments, and in the quantified 
version we need to verify if a first order sentence, whose quantifier-free 
part is conjunctive, is true in a given model. 

The general CSP is NP-complete \cite{Montanari74:constraints}. However, 
many practical and theoretical problems can be expressed in terms of CSPs 
using constraints of a certain restricted form. One of the most widely used 
way to restrict a constraint satisfaction problem is to specify the set of 
allowed constraints, which is usually a collection of relations on a finite set. 
The key result is that this set of relations can usually be assumed to be a 
co-clone of a certain kind. More precisely, a generic statement asserts that 
if a relation $\rel$ belongs to the co-clone generated by a set $\Gm$ of 
relations then the CSP over $\Gm\cup\{\rel\}$ is polynomial time reducible
 to the CSP over $\Gm$. Then we can use the appropriate Galois 
connection to transfer the question about sets of relations to a question 
about certain classes of functions.

For the classical decision CSP such a result was obtained by Jeavons et al. 
\cite{Jeavons97:closure}, who proved that intersection of relations (that is, 
conjunction of the corresponding predicates) and projections (that is, 
existential quantification) give rise to polynomial time reducibility of CSPs. 
Therefore in the study of the complexity of the CSP it suffices to focus on 
co-clones. Using the result of Geiger \cite{Geiger68:closed} or the one of  
Bodnarchuk et al.\ \cite{Bodnarchuk69:Galua1} one can instead consider 
clones of functions. A similar result is true for the counting CSP as shown 
by Bulatov and Dalmau \cite{Bulatov07:towards}. In the case of 
quantified CSP, B\"orner et al.\ proved \cite{Borner09:games} that 
conjunction, existential quantification, and also universal quantification 
give rise to a polynomial time reduction between quantified problems. The 
appropriate class of functions  is then the class of surjective functions. 
Along with the usual counting CSP, a version, in which one is required to 
approximate the number of solutions, has also been considered. The standard 
polynomial time reduction between problems is not suitable for approximation 
complexity. In this case, therefore, another type of reductions, approximation 
preserving, or, AP-reductions, is used. The first author proved in 
\cite{Bulatov09:clones} that conjunction of predicates gives rise to an 
AP-reduction between approximation counting CSPs. By the Galois connection 
established by Fleischner and Rosenberg \cite{Fleischer78:Galois}, the 
approximation  complexity of a counting CSP is a property of a clone of 
partial functions. 

In most cases establishing the connection between clones of functions 
and reductions between CSPs has led to a major success in the study of 
the CSP. For the decision problem, a number of very strong results have 
been proved using methods of universal algebra \cite{Bulatov05:classifying,%
Bulatov03:conservative,Bulatov06:3-element,Barto09:bounded-width,%
Idziak10:few}. For the exact counting CSP a complete complexity classification 
of such problems has been obtained \cite{Bulatov08:counting}. Substantial 
progress has been also made in the case of quantified CSP 
\cite{Chen08:collapsibility}.

Compared to the results cited above the progress made in the approximation 
counting CSP is modest. Perhaps, one reason for this is that clones of partial 
functions are much less studied, and much more diverse than clones of total 
functions. In this paper we attempt to overcome to some extent the difficulties 
arising from this weakness of partial clones. 

In the first part of the paper we introduce new types of quantification and 
show that such quantifications, we call them max-implementation and 
max-quantification, give rise to AP-reductions between approximation 
counting CSPs. Intuitively, applying the max-quantifier to a 
relation $\rel(\vc xn,y)$ results in the relation $\mexe y\rel(\vc xn,y)$ that 
contains those tuples $(\vc an)$ that have a maximal number of extensions 
$(\vc an,b)$ such that $\rel(\vc an,b)$ is satisfied. Max-implementation, $\mex$,
is a similar construction, but applied to a group of variables. Sets of relations closed 
with respect this new type of quantification will be called max-co-clones. 
Thus we strengthen the closure operator on sets of relation hoping that 
the sets of functions corresponding to the new type of Galois connection 
are easier to study. We were unable, however, to describe a Galois 
connection for sets closed under max-implementation and max-quantification. 
Instead, we consider a somewhat close type of quantifiers, $k$-existential 
quantifiers. Quantifiers of this type are known as counting quantifiers in 
model theory, and often used to enhance first order logic languages (see, e.g.\ 
\cite{Dawar07:counting}). Counting quantifiers are similar to max-existential 
quantifiers, although do not capture them completely. We call sets of relations 
closed under conjunctions and $k$-existential quantification $k$-existential 
co-clones.  On the functional side, an $n$-ary (partial) function on a set $D$ 
is said to be $k$-subset surjective if it is surjective on any collection of 
$k$-element subsets. More precisely, for any $k$-element subsets 
$\vc An\sse D$ the set $f(\vc An)$ contains at least $k$ 
elements. The second result of the paper asserts that $k$-existential 
co-clones are exactly the sets of relation invariant with respect to a set 
of $k$-subset surjective (partial) functions. Finally, we give a complete 
description of max-co-clones on $\{0,1\}$ (Boolean max-co-clones). Surprisingly,
any Boolean max-co-clone is also a usual co-clone (but not the other way 
around). We show that in general it is not true.

\section{Preliminaries}

By $[n]$ we denote the set $\{1\zd n\}$. For a set $D$, by $D^n$ we 
denote the set of all \emph{$n$-tuples} of
elements of $D$. An $n$-ary relation is any set $\rel\sse D^n$. The 
number $n$ is called the \emph{arity} of $\rel$ and denoted $\ar(R)$. Tuples
will be denoted in boldface, say, $\ba$, and their entries will be denoted by
$\ba[1]\zd \ba[n]$. For $I=(\vc ik)\sse[n]$ by $\pr_I\ba$ we denote the tuple
$(\ba[i_1]\zd\ba[i_k])$, and we use $\pr_I\rel$ to denote $\{\pr_I\ba\mid \ba\in\rel\}$. 
We will also need predicates corresponding to
relations. To simplify the notation we use the same symbol for a
relation and the corresponding predicate, for instance, for an $n$-ary
relation $\rel$ the corresponding predicate $\rel(\vc xn)$ is given by
$\rel(\ba[1]\zd \ba[n])=1$ if and only if $\ba\in\rel$. Relations and predicates
are used interchangeably.

For a set of relations $\Gm$ over a set $D$, the set $\dang\Gm$
includes all relations that can be expressed (as a predicate) using
(a) relations from $\Gm$, together with the binary equality relation
$=_D$ on $D$, (b) conjunctions, and (c) existential
quantification. This set is called the \emph{co-clone generated by} $\Gm$.

\emph{Partial co-clone generated by} $\Gm$ is obtained in a similar way
by disallowing existential quantification. $\ang\Gm$ includes all
relations that can be expressed using (a) relations from $\Gm$,
together with $=_D$, and (b) conjunctions,

If $\Gm=\ang\Gm$ or $\Gm=\dang\Gm$, the set $\Gm$ is said 
to be a \emph{partial co-clone}, and a \emph{co-clone}, respectively. 

Sometimes there is no need to apply even conjunction to produce a 
new relation. For instance, $\relo(x,y)=\rel(x,y,y)$ defines a 
binary relation from a ternary one. Therefore it is often 
convenient, especially for technical purposes, to group manipulations 
with variables of a relation 
into a separate category. More formally, for a relation $\rel(\vc xn)$
and a mapping $\pi\colon\{\vc xn\}\to V$, where $V$ is some
set of variables, $\pi\rel$ denotes the relation 
$\rel(\pi(x_1)\zd\pi(x_n))$. We will understand by (partial) 
co-clones sets of relations closed under manipulation with variables,
conjunction, and existential quantification (respectively, closed 
under manipulation with variables and conjunction).

Co-clones and partial co-clones can often be conveniently and 
concisely represented through functions and partial functions, respectively.

Let $\rel$ be a ($k$-ary) relation on a set $D$, and $f\colon D^n\to D$ an
$n$-ary function on the same set. Function $f$ \emph{preserves}
$\rel$, or is a \emph{polymorphism} of $\rel$, if for any $n$ tuples
$\vc \ba n\in\rel$ the tuple $f(\vc \ba n)$ obtained by component-wise
application of $f$ also belongs to $\rel$. Relation $\rel$ in this
case is said to be \emph{invariant} with respect to $f$. The set of
all functions that preserve every relation from a set of relations
$\Gm$ is denoted by $\Pol(\Gm)$, the set of all relations invariant
with respect to a set of functions $C$ is denoted by $\Inv(C)$.  

Operators $\Inv$ and $\Pol$ form a Galois connection between 
sets of functions and sets of relations. Sets of the form $\Inv(C)$ 
are precisely co-clones; on the functional side there is another type 
of closed sets.

A set of functions is said to be a \emph{clone} of functions if it
is closed under superpositions and contain all the \emph{projection}
functions, that is functions of the form $f(\vc xn)=x_i$.  Sets of
functions of the form $\Pol(\Gm)$ are exactly clones of
functions  \cite{Poschel-Kaluznin79} . 

The study of the \#CSP also makes use of another Galois connection, a
connection between partial co-clones and sets of \emph{partial
  functions}. An $n$-ary partial function $f$ on a set $D$ is just a
partial mapping $f\colon D^n\to D$. As in the case of total functions,
a partial function $f$ \emph{preserves} relation $\rel$, if for any
$n$ tuples $\vc \ba n\in\rel$ the tuple $f(\vc \ba n)$ obtained by
component-wise application of $f$ is either undefined or belongs to
$\rel$. The set of all partial functions that preserve every relation
from a set of relations $\Gm$ is denoted by $\PPol(\Gm)$.  

The set of all tuples from $D^n$ on which $f$ is defined is called the
\emph{domain} of $f$ and denoted by $\Dom(f)$. A set of functions is
said to be \emph{down-closed} if along with a function $f$ it
contains any function $f'$ such that $\Dom(f')\sse\Dom(f)$ and $f'(\vc
an)=f(\vc an)$ for every tuple $(\vc an)\in \Dom(f')$. A down-closed
set of functions, containing all projections and closed under
superpositions is called a \emph{partial clone}. Fleischner and Rosenberg 
\cite{Fleischer78:Galois} proved that partial clones are
exactly the sets of the form $\PPol(\Gm)$ for a certain $\Gm$, and that the 
partial co-clones are precisely the sets $\Inv(C)$ for collections $C$ 
of partial functions. 

\section{Approximate counting and max-implementation}

Let $D$ be a set, and let $\Gm$ be a finite set of relations over $D$. 
An instance of the counting Constraint Satisfaction Problem, 
$\NCSP(\Gm)$, is a pair $\cP=(V,\cC)$ where $V$ is a set of 
\emph{variables}, and $\cC$ is a set of \emph{constraints}. 
Every constraint is a pair $\ang{\bs,\rel}$, in which $\rel$ is a 
member of $\Gm$, and $\bs$ is a tuple of variables from $V$ 
of length $\ar(\rel)$ (possibly with repetitions). A \emph{solution} 
to $\cP$ is a mapping $\vf:V\to D$ such that $\vf(\bs)\in\rel$ 
for every constraint $\ang{\bs,\rel}\in\cC$. The objective in 
$\NCSP(\Gm)$ is to find the number $\#\cP$ of solutions to a 
given instance $\cP$. 

We are interested in the complexity of this problem depending 
on the set $\Gm$. The complexity of the exact counting problem 
(when we are required to find the exact number of solutions) is 
settled in \cite{Bulatov08:counting} by showing that for any finite 
$D$ and any set $\Gm$ of relations over $D$ the problem is 
polynomial time solvable or is complete in a natural complexity 
class $\#P$. One of the key steps in that line of research is the 
following result: For a relation $\rel$ and a set of relations $\Gm$ 
over $D$, if $\rel$ belongs to the co-clone generated by 
$\Gm$, then $\NCSP(\Gm\cup\{\rel\})$ is polynomial time reducible 
to $\NCSP(\Gm)$. This results emphasizes the importance of 
co-clones in the study of constraint problems. 

A situation is different when we are concerned about approximating 
the number of solutions. We will need some notation and terminology. 
Let $A$ be a counting problem. An algorithm $\Alg$ is said to be an 
\emph{approximation algorithm} for $A$ with relative error 
$\varepsilon$ (which may depend on the size of the input) if it 
is polynomial time and for any instance $\cP$ of $A$ it outputs a 
certain number $\Alg(\cP)$ such that $\Alg(\cP)=0$ if $\cP$ has
no solution and
$$
\frac{|\#\cP-\Alg(\cP)|}{\#\cP}<\varepsilon
$$
otherwise, where $\#\cP$ denotes the exact number of solutions to $\cP$. 

The following framework is viewed as one
of the most realistic models of efficient computations. A \emph{fully polynomial 
approximation scheme} (FPAS, for short) for a problem $A$ is an algorithm $\Alg$ 
such that: It takes as input an instance $\cP$ of $A$ and a real number 
$\varepsilon>0$, the relative error of $\Alg$ on the input $(\cP,\varepsilon)$ is 
less than $\varepsilon$, and $\Alg$
is polynomial time in the size of $\cP$ and $\log(\frac1\varepsilon)$. 

To determine the approximation complexity of problems approximation preserving 
of reductions are used. Suppose $A$ and $B$ are two counting problems 
whose complexity (of approximation) we want to compare. An 
\emph{approximation preserving reduction} or \emph{AP-reduction} from 
$A$ to $B$ is an algorithm $\Alg$, using $B$ as an oracle, that 
takes as input a pair $(\cP,\varepsilon)$ where $\cP$ is an instance of $A$ 
and $0<\varepsilon<1$, and satisfies the following three conditions: (i) every 
oracle call made by $\Alg$ is of the form $(\cP',\dl)$, where $\cP'$ is an 
instance of $B$, and $0<\dl<1$ is an error bound such that 
$\log\left(\frac1\dl\right)$ is bounded by a polynomial in the size of 
$\cP$ and $\log\left(\frac1\varepsilon\right)$; 
(ii) the algorithm $\Alg$ meets the specifications for being  
an FPAS for $A$ whenever the oracle meets the specification 
for being an FPAS for $B$; and (iii) the running 
time of $\Alg$ is polynomial in the size of $\cP$ and $\log(\frac1\varepsilon)$. 
If an approximation preserving reduction from $A$ to $B$ exists we denote 
it by $A\le_{\rm AP} B$, and say that $A$ is \emph{AP-reducible to} $B$.

Similar to co-clones and polynomial time reductions, partial co-clones 
can be shown to be preserved by AP-reductions.

\begin{theorem}[\cite{Bulatov09:clones}]\label{the:partial}
Let $\rel$ be a relation and $\Gm$ be a set of relations over a finite set
such that  $\rel$ belongs to
$\ang{\Gm}$. Then $\NCSP(\Gm\cup\{\rel\})$ is AP-reducible to $\NCSP(\Gm)$.  
\end{theorem}

This result however has two significant setbacks. First, partial co-clones 
are not studied to the same extent as regular co-clones, and, due to 
greater diversity, are not believed to be ever studied to a comparable 
level. Second, it does not used the full power of AP-reductions, and 
therefore leaves significant space for improvements. In the rest of 
this section we try to improve upon the second issue.

\begin{defin}\label{def:max-implementation}
Let $\Gm$ be a set of relations on a set $D$, and let $\rel$ be an 
$n$-ary relation on $D$. Let $\cP$ be an instance of $\NCSP(\Gm)$ 
over the set of variables consisting of $V=V_x\cup V_y$, where 
$V_x=\{x_1,x_2,\cdots,x_n\}$ and $V_y=\{y_1,y_2,\cdots,y_q\}$. 
For any assignment of $\vf:V_x\to D$, let $\#\varphi$ be the number 
of assignments $\psi:V_y\to D$ such that $\varphi\cup \psi$ 
satisfy $\cP$. Let $M$ be the maximum value of $\#\varphi$ among 
all assignments  of $V_x$. The instance $\cP$ is said to be a 
\emph{max-implementation} of $\rel$ if a tuple $\varphi$ is in $R$ 
if and only if $\#\varphi=M$.
\end{defin}


\begin{theorem}\label{the:max-implementation}
If there is max-implementation of $\rel$ by  $\Gm$, then 
$\NCSP(\Gm \cup \{\rel\})\le_{AP}\NCSP(\Gm)$.
\end{theorem}

\begin{proof}
Let $\cP=(V=V_x\cup V_y,\cC)$ be a max-implementation of $\rel$ 
by $\Gm$, and let $M$ be the maximal number of extensions of 
assignments of $V_x$ to solutions of $\cP$. For any instance 
$\cP_1=(V_1,\cC_1)$ of $\NCSP(\Gm\cup \{\rel\})$ we construct an 
instance $\cP_2=(V_2,\cC_2)$ of $\NCSP(\Gm)$ as follows.
\begin{itemize}
\item
Choose a sufficiently large integer $m$ (to be determined later).
\item
Let $\vc C\ell\in\cC_1$ be the constraints from $\cP_1$ involving 
$\rel$, $C_i=\ang{\bs_i,\rel}$. Set 
$V_2=V_1\cup\bigcup_{i=1}^\ell(V^i_1\cup\ldots\cup V^i_m)$, 
where each $V^i_j$ is a fresh copy of $V_y$. 
\item
Let $\cC$ be the set of constraints of $\cP$. Set 
$\cC_2=(\cC_1-\{\vc C\ell\}) \cup
\bigcup_{i=1}^\ell(\cC^i_1\cup\ldots\cup \cC^i_m)$, where each 
$\cC^i_j$ is a copy of $\cC$ defined as follows. For each 
$\ang{\bs,\relo}\in\cC$ we include $\ang{\bs^i_j,\relo}$ into 
$\cC^i_j$, where $\bs^i_j$ is obtained from $\bs$ replacing every 
variable from $V_y$ with its copy from $V^i_j$. 
\end{itemize} 

Now, as is easily seen, every solution of $\cP_1$ can be extended 
to a solution of $\cP_2$ in $M^{\ell m}$ ways. Observe that 
sometimes the restriction of a solution $\psi$ of $\cP_2$ to $V_1$ 
is not a solution of $\cP_1$. Indeed, it may happen that although 
$\psi$ satisfies every copy $\cC^i_j$ of $\cP$, its restriction to 
$\bs^i_j$ does not belong to $\rel$, simply because this restriction 
does not have sufficiently many extensions to solutions of $\cP$. 
However, any assignment to $V_1$ that is not a solution to $\cP_1$ 
can be extended to a solution of $\cP_2$ in at most 
$(M-1)^m\cdot M^{(\ell-1)m}$ ways. Hence,
$$
M^{\ell m} \cdot \#\cP_1\le \#\cP_2  
\le M^{\ell m} \cdot \#\cP_1 + |V_1|^{|D|} \cdot (M-1)^m\cdot M^{(\ell-1)m}
$$
Then we output $\frac{\#\cP_2}{M^{\ell m}}$.

Let $|V_1|=k$ and $|D|=d$. Given a desired relative error 
$\varepsilon$ we have to find $m$ such that 
$$
\frac{\frac{\#\cP_2}{M^{\ell m}}-\#\cP_1}{\#\cP_1}<\varepsilon.
$$
A straightforward computation shows that any
$$
m> \frac{d\log k-\log\varepsilon}{\log(M-1)-\log M}
$$
achieves the goal.
\end{proof}

Max-implementation can be used as another closure operator on 
the set of relations. Let $\rel(\vc xn,\vc ym)$ be a relation on a set $D$. By 
$\mex(\vc ym)\rel(\vc xn,\vc ym)$ we denote the relation 
$\relo(\vc xn)$ on the same set given by the rule: $\ba\in\relo$ if and 
only if there are $M$ tuples $\bb\in D^m$ such that $(\ba,\bb)\in\rel$,
where $M$ is the maximal number of elements in the set 
$\{\bb\mid (\ba,\bb)\in\relo\}$ over all $\ba\in D^n$.
A set of relations $\Gm$ over $D$ is said to 
be a \emph{max-co-clone} if it contains the equality relations, and 
closed under conjunctions and max-implementations. The smallest 
max-co-clone containing a set of relations $\Gm$ is called the 
\emph{max-co-clone generated by} $\Gm$ and denoted $\mang\Gm$.

\begin{lemma}\label{lem:max-clones-implementations}
Let $\Gm$ be a set of relations and $\rel\in\mang\Gm$. Then there is 
a max-implementation of $\rel$ by $\Gm$.
\end{lemma}

\begin{proof}
Suppose $\rel\in\mang\Gm$. We need to show that $\rel$ can be represented 
as $\rel(\vc xn)=\mex(\vc ym)$\lb$ \Phi(\vc xn,\vc ym)$, where $\Phi$ is quantifier
free. To this end it suffices to prove three equalities:
\begin{enumerate}
\item
if $\rel(\vc xn)=\mex(\vc ym)\Phi(\vc xn,\vc ym)$ and $\pi$ is a transformation 
of the set $\{\vc xn\}$ then 
$(\pi\rel)(\vc xn)=\mex(\vc ym)\Phi(\pi(x_1)\zd\pi(x_n),\vc ym)$;
\item
if $\rel(\vc xn)=\mex(\vc ym)\Phi_1(\vc xn,\vc ym)\meet \mex(\vc zr)\Phi_2(\vc xn,\vc zr)$,
then $\rel(\vc xn)=\mex(\vc ym,\vc zr)(\Phi_1(\vc xn,\vc ym)\meet \Phi_2(\vc xn,\vc zr))$;
\item
if $\rel(\vc xn)=\mex(\vc ym)\mex(\vc zr)\Phi(\vc xn,\vc ym,\vc zr)$, then
there is a quantifier free formula $\Psi$ such that 
$\rel(\vc xn)=\mex(\vc us)\Psi(\vc xn,\vc us)$.
\end{enumerate}

(1) follows straightforwardly from definitions.

(2) $\ba\in\rel$ if and only if it has the maximal number of extensions in
both $\Phi_1$ and $\Phi_2$. Without loss of generality, sets $\{\vc ym\}$ and
$\{\vc zr\}$ are disjoint. Let a tuple $\ba\in\rel$ have $M_1$ extensions in 
$\Phi_1$ and $M_2$ extensions in $\Phi_2$. Then it has $M_1M_2$ extensions
in $\Phi_1\meet\Phi_2$. On the other hand, let $\ba\not\in\rel$. Let also
it have $M'_1$ extensions in $\Phi_1$ and $M'_2$ extensions in $\Phi_2$, and
either $M'_1< M_1$ or $M'_2<M_2$. Since such tuple has $M'_1M'_2<M_1M_2$ extensions,
it does not belong to the relation defined by 
$\rel(\vc xn)=\mex(\vc ym,\vc zr)(\Phi_1(\vc xn,\vc ym)\meet \Phi_2(\vc xn,\vc zr))$
as well.

(3) Observe first that $\rel(\vc xn)$ does not necessarily equal \lb 
$\mex(\vc ym,\vc zr)\Phi(\vc xn,\vc ym,\vc zr)$. Indeed, let $\Phi'$ denote the 
formula\lb 
$\relo(\vc xn,\vc ym)=\ex(\vc zr)\Phi(\vc xn,\vc ym,\vc zr)$. Then it is 
possible that although every extension of a tuple $\ba$ to $(\ba,\bb)\in\relo$ 
has very few extensions to a tuple from $\Phi$, and so $\ba\not\in\rel$, the number 
of extensions $\bb$ is large so that combined $\ba$ has enough extensions to
tuples from $\Phi$.

To avoid this we make sure that extensions to tuples from $\relo$ cannot make up 
for extensions to $\Phi$. Let $M$ be the maximal number of extensions $\bb$ of 
tuple $\ba$ such that $(\ba,\bb)\in\relo$, and $N$ the maximal number of 
extensions $\bc$ of $(\ba,\bb)\in\relo$ to $(\ba,\bb,\bc)\in\Phi$. Let also $L$ 
be the maximal number of extensions $\bb$ of $\ba\in\rel$; it is possible 
that $L<M$. Set 
$$
c=\max\left(1,\left\lceil\log\frac LM/ \log\frac {N-1}N\right\rceil\right).
$$ 
We show that 
$\rel(\vc xn)=\mex(\vc us)\Psi(\vc xn,\vc us)$, where $\{\vc us\}=$\lb$
\{\vc ym,z_1^1\zd z_r^1\zd z_1^c\zd z_r^c\}$, and 
$$
\Psi(\vc xn,\vc us)=\bigwedge_{s=1}^c \Phi(\vc xn,\vc ym,z_1^s\zd z_r^s).
$$ 
If a tuple $\ba$ belongs to $\rel$ it is extendable in $L$ ways to a tuple from
$\relo$, and then every such extended tuple $(\ba,\bb)$ is extendable in $N$
ways to a tuple from $\Phi$. Therefore $\ba$ has $LN^c$ extensions to a tuple 
from $\Psi$. On the other hand, if $\ba\not\in\rel$, then it can be extended 
in at most $M$ ways to a tuple $(\ba,\bb)\in\relo$, then this tuple is extendable 
in at most $N-1$ ways to a tuple from $\Phi$. Thus $\ba\not\in\rel$ has 
$$
M(N-1)^c =LN^c\cdot\frac ML\left(\frac{N-1}N\right)^c <LN^c
$$
extensions.
\end{proof}

The next natural step would be to find a type of func\-tions and a 
closure operator on the set of functions that give rise to a Galois 
connection capturing max-co-clones. 

\begin{problem}\label{pro:max-implementation}
Find a class $\cF$ of (partial) functions and a closure operator 
$[\cdot]$ on this class such that for any set of relations $\Gm$ 
and any set $C\sse\cF$ it holds that 
$\mang\Gm=\Inv(\cF\cap\PPol(\Gm))$, and $[C]=\cF\cap\PPol\,\Inv(C)$.
\end{problem}

In all the cases previously studied the projection (or quantification) type 
operators on relations can be reduced to quantifying away a 
single variable. However, max-implementations seem to inherently 
involve  a number of variables, rather than a single variable. 
In the end of this paper we use our description of Boolean 
max-co-clones to show that max-implementations are provably more
powerful than max-quantification (see below). In the Boolean case every
max-quantification is equivalent to either existential quantification,
or universal quantification. Sets of relations on $\{0,1\}$ closed
under these two types of quantifications are well known: these
are sets of invariant relations of sets of surjective functions 
\cite{Borner09:games}. However, not all of them are max-co-clones.

Therefore a meaningful relaxation of max-co-clones restricts 
the use of max-implementation to one auxiliary variable. 
Let $\Phi$ be a formula with free variables $\vc xn$ and $y$ over 
set $D$ and some predicate symbols. Then $\vc an$ satisfy 
$$
\Psi(\vc xn)=\mexe y\Phi(\vc xn,y)
$$
if and only if the number of $b\in D$ such that $\Phi(\vc an,b)$ 
is true is maximal among all tuples $(\vc cn)\in D^n$. The 
quantifier $\mexe$ will be called \emph{max-quantifier}. 
A set of relations $\Gm$ over $D$ is said to be a 
\emph{max-existential co-clone} if it contains the equality relation, 
and closed under conjunctions and max-existential quantification. 
The smallest max-existential co-clone containing a set of relations 
$\Gm$ is called the \emph{max-existential co-clone generated by} 
$\Gm$ and denoted $\mange\Gm$.

 \begin{problem}\label{pro:max-quantification}
Find a class $\cF$ of (partial) functions and a closure operator 
$[\cdot]$ on this class such that for any set of relations $\Gm$ 
and any set of functions $C\sse\cF$ it holds that 
$\mange\Gm=\Inv(\cF\cap\PPol(\Gm))$, and  $[C]=\cF\cap\PPol\,\Inv(C)$.
\end{problem}

In the next section we consider certain constructions approximating max-existential co-clones.

\section{$k$-Existential and max-existential co-clones}

In order to approach max-quantification we consider counting quantifiers that have 
been used in model theory to increase the power of first order logic 
\cite{Immerman90:describing,Etessami97:counting}. 

Let $\Phi$ be a formula with free variables $\vc xn$ and $y$ over set $D$ and 
some predicate symbols. Then $\vc an$ satisfy 
$$
\Psi(\vc xn)=\kex y\Phi(\vc xn,y)
$$
if and only if $\Phi(\vc an,b)$ is true for at least $k$ values $b\in D$. The 
quantifier $\kex$ will be called \emph{$k$-existential quantifier}. 
It is easy to see that 1-existential quantifier is just the regular existential 
quantifier, and the $|D|$-existential quantifier is equivalent to the universal 
quantifier on set $D$.  

We now introduce several types of co-clones depending on what kind of 
$k$-existential quantifiers are allowed. A set of relations $\Gm$ over set 
$D$ is said to be a \emph{$k$-existential partial co-clone} if it contains the 
equality relation $=_D$, and closed under manipulations with variables, 
conjunction, and $k$-existential quantification. The smallest $k$-existential partial 
co-clone containing a set of relations $\Gm$ is called the 
\emph{$k$-existential partial co-clone generated by} $\Gm$ and denoted 
$\ang\Gm_k$. In a similar way we can define sets of
relations closed under several counting quantifiers. Let $K\sse\nat$. 
A set of relations $\Gm$ over set 
$D$ is said to be a \emph{$K$-existential partial co-clone} if it contains the 
equality relation $=_D$, and closed under manipulations with variables, 
conjunction, and $k$-existential quantification for $k\in K$. Clearly, if $\Gm$ is a set of relations 
on an $m$-element set, we may assume $K\sse[m]$. 
If $1\in K$, set $\Gm$ is closed under existential quantification, and so it is called 
a \emph{$K$-existential co-clone}. If, in addition, $K=\{1,k\}$, $\Gm$ is called 
\emph{$k$-existential co-clone}. The set $\Gm$ is said to be a 
\emph{counting co-clone}\footnote{`Counting' in this term comes from 
counting quantifiers and has nothing to do with counting constraint 
satisfaction.} 
if it is an $\nat$-existential partial co-clone, that is, if it contains $=_D$, and 
closed under conjunctions and $k$-existential 
quantification for all $k\ge 1$. The smallest $K$-existential partial co-clone 
($K$-existential co-clone, $k$-existential co-clone, counting co-clone)
 containing $\Gm$ are called the 
\emph{$K$-existential partial co-clone} (\emph{$K$-existential co-clone}, 
\emph{$k$-existential co-clone}, \emph{counting co-clone}) 
\emph{generated by} $\Gm$ and denoted $\ang\Gm_K$ ($\dang\Gm_K$,
$\dang\Gm_k$, $\dang\Gm_\infty$, respectively). 

We observe some simple properties of counting quantifiers. 

\begin{lemma}\label{lem:counting-quantifiers-property}
Let $\Phi(\vc xn,\vc ym)$ and $\Psi(\vc xn,\vc z\ell)$ be conjunctive 
quantifier free formulas. Then
\begin{eqnarray*}
\lefteqn{\ex_{s_1}y_1\ldots\ex_{s_m}y_m\ex_{t_1}z_1\ldots\ex_{t_\ell}
\, (\Phi(\vc xn,\vc ym)\meet\Psi(\vc xn,\vc z\ell))}\\
&=&
(\ex_{s_1}y_1\ldots\ex_{s_m}y_m
\, (\Phi(\vc xn,\vc ym))\meet(\ex_{t_1}z_1\ldots\ex_{t_\ell}\,\Psi(\vc xn,\vc z\ell)),
\end{eqnarray*}
for any $\vc sm,\vc t\ell\in\nat$, provided $\vc ym,\vc z\ell\not\in\{\vc xn\}$
and $\{\vc ym\}\cap\{\vc z\ell\}=\eps$.
\end{lemma}

\begin{corollary}\label{cor:prenex}
Let $\Gm$ be a set of relations on a set $D$, $K\sse\nat$, 
and $\rel(\vc xn)\in\ang\Gm_K$. Then there is a conjunctive quantifier free
formula $\Phi(\vc xn,\vc ym)$ using relations from $\Gm$ and the equality relation such that
$$
\rel(\vc xn)=\ex_{s_1}\ldots\ex_{s_m}\, \Phi(\vc xn,\vc ym).
$$
\end{corollary}

The following observation summarizes some relationship between 
the constructions introduced.

\begin{observation}\label{obs:inclusions}
For a set of relations $\Gm$ on $D$, $|D|=m$, the following hold.\\[.5mm]
\indent
- $\Gm$ is a 1-existential (partial) co-clone if and only if it is a co-clone.\\[.5mm]
\indent
- $\Gm$ is a (partial) $m$-existential clone if and only if it is a (partial) co-clone closed under universal quantification.\\[.5mm]
\indent
- if $\Gm$ is a counting co-clone then it is a max-existential co-clone.\\[.5mm]
\indent
- if $\Gm$ is a max-existential co-clone then it is a partial $m$-existential co-clone.
\end{observation}

In all other cases the introduced versions of co-clones are incomparable.

\begin{example}\label{exa:incomparability}\rm
Fix a natural number $m$ and let $D$ be a set with $\frac{m(m-1)}2$ 
elements. Consider an equivalence relation $\rel_m$ on $D$ with 
classes $D_1\zd D_m$ such that $|D_i|=i$. Then the co-clone generated 
by $\rel_m$ corresponds to one of the Rosenberg's maximal clones 
\cite{Rosenberg70:funktionale}, and so the structure of relations from 
this co-clone is well understood. For any $n$-ary relation 
$\relo\in\dang{\rel_m}$ there is a partition $\vc Ik$ of $[n]$ such that 
a tuple $\ba$ belongs to $\relo$ if and only if for each $j\in[k]$ and every 
$i,i'\in I_j$ the entries $\ba[i],\ba[i']$ are $\rel_m$-related. This also 
means that $\ang{\rel_m}=\dang{\rel_m}$.

Applying $k$-existential and max-existential quantifiers one can easily 
find the $k$-existential, counting, and max-existential clones generated 
by $\rel$:
\begin{enumerate}
\item
$\ang{\rel_m}_k=\dang{\rel_m}_k$ is the set of relations $\relo$:
There is a partition $\vc It$ of $[\ar(\relo)]$ and $J\sse[t]$ such 
that a tuple $\ba$ belongs to $\relo$ if and only if for each $j\in[t]$ 
and every $i,i'\in I_j$ the entries $\ba[i],\ba[i']$ are $\rel_m$-related 
and $\ba[i]\in D_k\cup\ldots\cup D_m$ for $i\in I_j$, $j\in J$. 
\item
$\dang{\rel_m}_\infty$ is the set of relations $\relo$:
There is a partition $\vc It$ of $[\ar(\relo)]$ and a function $\vf:[t]\to[m]$ 
such that a tuple $\ba$ belongs to $\relo$ if and only if for each $j\in[t]$ 
and every $i,i'\in I_j$ the entries $\ba[i],\ba[i']$ are $\rel_m$-related 
and $\ba[i]\in D_{\vf(j)}\cup\ldots\cup D_m$ for $i\in I_j$, $j\in J$. 
\item
$\mang{\rel_m}=\mange{\rel_m}$ is the set of relations $\relo$:
There is a partition $\vc It$ of $[\ar(\relo)]$ and $J\sse[t]$ such that 
a tuple $\ba$ belongs to $\relo$ if and only if for each $j\in[t]$ and 
every $i,i'\in I_j$ the entries $\ba[i],\ba[i']$ are $\rel_m$-related 
and $\ba[i]\in D_m$ for $i\in I_j$, $j\in J$. 
\end{enumerate}

A set $\Gm$ such that $\ang\Gm_k\ne\dang\Gm_k$ can be easily 
found among usual weak co-clones. For instance, for any weak co-clone
$\Gm$ that is not a co-clone we have $\ang\Gm_1\ne\dang\Gm_1$. 
Such a weak co-clone can be found in, say, \cite{Haddad95:permutations}. 

In the example given we have $\mange{\rel_m}=\ang{\rel_m}_m$. 
However, since $\ang{\rel_{m-1}}_m=\ang{\rel_{m-1}}$, we have 
$\mange{\rel_{m-1}}\ne\ang{\rel_{m-1}}_m$. 
For an example distinguishing between $\mang\Gm$ and $\mange\Gm$ 
see the Conclusion.

We give a sketchy proof of (1) here, the remaining results are similar.
Let $\relo(\vc xn)$ satisfies the conditions in (1) for a partition $\vc It$ of
$[n]$ and $J\sse[t]$. Without loss of generality assume $J=[s]$, $s\le t$.
Choose variables $\vc ys\not\in\{\vc xn\}$ and consider relation 
$\rela(\vc xn,\vc ys)$ given by: $\ba\in\rela$ if and only if 
$(\ba[i],\ba[j])\in\rel_m$ for any 
$i,j\in I_\ell$ for some $\ell\in[t]$ and $(\ba[i],\ba[n+\ell])\in\rel_m$ 
for any $i\in I_\ell$ where $\ell\in J$. Clearly,
$\rela\in\ang{\rel_m}=\dang{\rel_m}$. Now, as it is easy to see, 
$$
\relo(\vc xn)=\kex y_1\ldots\kex y_s\,\rela(\vc xn,\vc ys).
$$

In order to show that every relation from $\dang{\rel_m}_k$ satisfies these conditions,
it suffices to prove that the set of relations $\Gm$ satisfying them is closed under 
manipulations with variables, conjunction, existential quantification, and
$k$-existential quantification. The first three operations are easy, since
$\Gm$ is a co-clone generated by $\rel_m$ and unary relation 
$D'=D_k\cup\ldots\cup D_m$. Let $\relo(\vc xn)\in\Gm$ and 
$\rela(\vc x{n-1})=\kex x_n\,\relo(\vc xn)$. Let also $\vc It$ and $J\sse[t]$
be the partition and a set from conditions (1). We may assume $n\in I_t$. 
Then if $t\in J$ then $\rela(\vc x{n-1})=\ex x_n\,\relo(\vc xn)$. Otherwise
$\ba\in\rela$ if and only if (a) for any $i,j\in I_\ell$, $\ell<t$, we have 
$(\ba[i],\ba[j])\in\rel_m$, (b) for any $i,j\in I'_t=I_t-\{n\}$, we have 
$(\ba[i],\ba[j])\in\rel_m$, and (c) $\ba[i]\in D'$, whenever 
$i\in I'_t\cup\bigcup_{s\in J}I_s$. Therefore $\rela\in\dang{\rel_m}_k$.
\end{example}

\section{Galois correspondence}

Let $D$ be a finite set. A (partial) function $f\colon D^n\to D$ is 
said to be \emph{$k$-subset surjective} if for any $k$-element 
subsets $\vc An\sse D$ the image $f(A_1\zd A_n)$ has cardinality 
at least $k$. A (partial) function that is $k$-subset surjective for 
each $k$, $1\le k \le|D|$ is said to be \emph{subset surjective}. 
The set of all arity $n$ $k$-subset surjective partial functions [arity 
$n$ $k$-subset surjective functions, subset surjective functions] on 
$D$ will be denoted by $\mkP n$ [resp., $\mkF n$, $\mF n$]; 
furthermore, $\mkPa=\bigcup_{n\ge0}\mkP n$, 
$\mkFa=\bigcup_{n\ge0}\mkF n$, $\mFa=\bigcup_{n\ge0}\mF n$.  
Any partial function is 1-subset surjective, while $|D|$-subset 
surjective partial functions are exactly the surjective partial functions. 
Observe that this definition can 
be strengthened by allowing the sets $A_i$, $i\in [n]$, to have at 
least $k$ elements.

\begin{lemma}\label{lem:relax-subset-monotone}
If an $n$-ary function $f$ is $k$-subset surjective, then for any subsets 
$\vc An\sse D$ with $|A_i|\ge k$, $i\in[n]$, the image $f(A_1\zd A_n)$ 
has cardinality at least $k$.
\end{lemma}

\begin{proof}
Choose any $B_i\sse A_i$, 
$i\in[n]$, and set $B=f(B_1\zd B_n)$. As $f$ is 
$k$-subset surjective, $|B|\ge k$. Finally, $B\sse f(A_1\zd A_n)$, and 
the result follows.
\end{proof}

The conditions of being $k$-subset surjective for different $k$ are 
in general incomparable, as the following example shows.

\begin{example}\label{exa:incomparable}\rm
Let $D=\{0\zd k-1\}$ be a $k$-element set and $1<m\le k$. Then the 
following function $f$ is not $m$-subset surjective, but is $\ell$-subset surjective 
for any $\ell\in[k]$ except $\ell=m$. Function $f$ is binary and given by its 
operation table:
$$
\left(\begin{array}{cccccccc}
0&0&\cdots&0&1&m&\cdots&k-1\\
1&1&\cdots&1&2&m&\cdots&k-1\\
\vdots&\vdots& &\vdots&\vdots&\vdots& &\vdots\\
m-3&m-3&\cdots&m-3&m-2&m&\cdots&k-1\\
m-2&m-2&\cdots&m-2&0&m&\cdots&k-1\\
0&1&\cdots&m-2&0&m&\cdots&k-1\\
0&1&\cdots&m-2&m-1&m&\cdots&k-1\\
\vdots&\vdots& &\vdots&\vdots&\vdots& &\vdots\\
0&1&\cdots&m-2&m-1&m&\cdots&k-1\\
\end{array}\right).
$$
Clearly, $f$ is not $m$-subset surjective, because $f(B,B)=\{0\zd m-2\}$ for 
$B=\{0\zd m-1\}$. Also, as it is a total function, $f$ is 1-subset surjective.
Take $\ell\in[k]$, $\ell>1$, and $B_1,B_2\sse\{0\zd k-1\}$ with
$|B_1|=|B_2|=\ell$. If there is $a\in B_1$ with $i\ge m$ then 
$f(a,b_1)\ne f(a,b_2)$ whenever $b_1\ne b_2$. This means that 
$|f(B_1,B_2)|\ge\ell$ 
in this case, and, in particular, $f$ is $\ell$-subset surjective for any $\ell>m$.
So, suppose $\ell< m$ and $B_1\sse\{0\zd m-1\}$
If $B_1\sse\{0\zd m-2\}$ then take $b\in B_2\cap\{0\zd m-2\}$ and observe that
$f(a_1,b)\ne f(a_2,b)$ for any $a_1,a_2\in \{0\zd m-2\}$, $a_1\ne a_2$.
Thus, $|f(B_1,\{b\})|=\ell$. Suppose $m-1\in B_1$.
If $B_2\sse\{0\zd m-2\}$, then $|f(m-1,B_2)|=\ell$; assume $m-1\in B_2$.
As is easily seen, $B_1\cap\{0\zd m-2\}\sse f(B_1,B_2)$. 
There is  $a\in \{0\zd m-2\}$ such that $a\not\in B_1$ but 
$a-1\pmod{m-1}\in B_1$. Then $a\in f(B_1,B_2)$, since $a=f(a-1,m-1)$.
Thus, $|f(B_1,B_2)|\ge\ell$.
\end{example}

The notion of invariance for $k$-subset surjective functions is the 
standard one for partial functions and relations. 
As usual, if $C$ is a set of ($k$-) subset surjective (partial) functions, 
$\Inv(C)$ denotes the set of relations invariant with respect to every 
function from $C$. For a set $\Gm$ of relations, $\mkPol(\Gm)$ and  
$\mkpPol(\Gm)$ denote the set of all $k$-subset surjective functions 
and partial functions, respectively, preserving every relation from $\Gm$.
For a set $K\sse\nat$ by $\mKPol(\Gm)$ and $\mKpPol(\Gm)$ we 
denote the set of all functions and, respectively, partial functions 
preserving every relation from $\Gm$ that are $k$-subset surjective 
for each $k\in K$. Thus, in particular, 
$$
\mKPol(\Gm)=\bigcap_{k\in K}\mkPol(\Gm),\quad\text{and}\quad
\mKpPol(\Gm)=\bigcap_{k\in K}\mkpPol(\Gm).
$$
By $\mPol(\Gm)$ we denote the analogous set of subset surjective functions.

The operator $\Inv$ on one side and the operators $\mkpPol(\Gm)$, 
$\mkPol(\Gm)$, $\mKPol(\Gm)$, $\mpPol(\Gm)$, $\mPol(\Gm)$ on 
the other side form Galois correspondences in the standard fashion. 
We characterize closed sets of relations that give rise from this correspondence.

\begin{lemma}\label{lem:quantifier preservation}
Let $\rel(\vc x\ell,y)$ be a relation on $D$, and let $\relo(\vc x\ell)=
\kex y \rel(\vc x\ell,y)$. Then if a $k$-subset surjective (partial) function 
$f$ preserves $\rel$, it also preserves $\relo$.
\end{lemma}

\begin{proof}
Suppose $f$ is $n$-ary. Take $\vc\ba n\in\relo$. Since each of them 
is put into $\relo$ by $k$-existential quantification, it has at least 
$k$ extensions to a tuple from $\rel$. Let $\vc Bn\sse D$ be such 
that $|B_i|\ge k$ and $(\ba_i,b)\in\rel$ for $b\in B_i$ and $i\in[n]$. 
Let also $\bb=f(\vc\ba n)$. For any $b\in B=f(B_1\zd B_n)$ the tuple 
$(\bb,b)$ belongs to $\rel$. As $f$ is $k$-subset surjective, $|B|\ge k$, 
hence, $\bb\in\relo$.
\end{proof}
 
\begin{theorem}\label{the:rel-clones}
Let $\Gm$ be a set of relations on a set $D$ and $K\sse\nat$. Then
$\Inv(\mKpPol(\Gm))=\ang\Gm_K$.
\end{theorem}

\begin{proof}
We will assume that $K=\{\vc ks\}\sse\{1\zd|D|\}$. Indeed, if $k\ge|D|$ then 
$\kex x\rel$ is empty for any relation on $D$. The equality relation, 
$=_D$, is invariant with respect to any 
partial function on $D$. Let $f$ be a $k$-subset surjective 
functions. It is straightforward to verify that manipulations of 
variables of a predicate invariant under $f$ and the conjunction of 
any two predicates invariant under $f$ result in predicates 
invariant under $f$, again, since it is true for any partial function. 
By Lemma~\ref{lem:quantifier preservation}
applying $k$-quantification to a predicate invariant under $f$ 
gives a predicate invariant under $f$, again because it is true
for any partial function. Hence, 
$\ang\Gm_K\sse\Inv\left(\mKpPol(\Gm)\right)$. Moreover, it 
follows that $\Inv\left(\mKpPol(\Gm)\right)=
\Inv\left(\mKpPol(\ang\Gm_K)\right)$.

To establish the reverse inclusion, take an $\ell$-ary relation 
$\rel\in\Inv(\mKpPol(\Gm))$. We need to show that $\rel\in\ang\Gm_k$. 
Define a relation $\relo$ as 
follows. Let $\rel=\{\vc\ba t\}$. For each $k\in K$ we consider sequences $(\vc Bt)$ 
of $k$-element subsets of $D$. Let also $(B^{k1}_1\zd B^{k1}_t)\zd 
(B^{kr_k}_1\zd B^{kr_k}_t)$ be a list of all such sequences. Let $\rela^j_k$
be the relation 
$$
\underbrace{B^{k1}_j\tm\ldots\tm B^{k1}_j}_{k\text{ times}}
\tm\ldots\tm \underbrace{B^{kr_k}_j\tm\ldots\tm B^{kr_k}_j}_{k\text{ times}},
$$ 
and $\rela^j=\rela^j_{k_1}\tm\ldots\tm\rela^j_{k_s}$.  Then $\relo$ 
is the union of relations given by
$
\ba_j\tm\rela^j,
$
for all $j\in[t]$. We show that there is $\rela\in\ang\Gm_k$ such that 
$\relo\sse\rela$ and $\pr_{[\ell]}\rela=\rel$. Then applying 
$k$-quantifications, $k\in K$, to all coordinates of $\rela$ except for the first 
$\ell$ we infer that $\rel\in\ang\Gm_K$. 

Set $M=\sum_{k\in K}kr_k$ and $M_j=\sum_{i=1}^j k_ir_{k_i}$; 
by $N_K$, $k\in K$, we denote the set $\{M_j+1\zd M_{j+1}\}$ . 
Let us consider the relation $\rela=\bigcap\{\relo'\in\ang\Gm_K\mid 
\relo\sse\relo'\}$. Since $\ang\Gm_K$ is closed under conjunctions 
and contains the total relation $D^{\ell+M}$, we have 
$\rela\in\ang\Gm_K$ and $\relo\sse\rela$. 

Now choose any tuple $\bb=(b_1\zd b_\ell,d_1\zd d_{M})\in\rela$. 
There are sets $\vc C{M}$ such that $|C_i|=k_j$, $i\in[M]$, whenever
$i\in N_j$, for any $t\in[r_j]$, $C_{M_{j-1}+k_j(t-1)+1}=\ldots=
C_{M_{j-1}+k_jt}$, $d_i\in C_i$, and for any 
$d'_i\in C_i$, $i\in[M]$, the tuple 
$(b_1\zd b_\ell,d'_1\zd d'_{M})\in\rela$. Indeed, otherwise we can 
applying a sequence of $k$-quantifications for $k\in K$ to obtain an 
$\ell$-ary relation $\rela'$ containing $\rel$, but not $(\vc b\ell)$. 
Then , $(\rela'\tm D^{\ell+M})\cap\relo$ belongs to $\ang\Gm_K$,
but is smaller than $\relo$. Therefore we can choose $\bb$ such 
that for any $j\in [s]$ and any $t\in [r_j]$ all the values 
$d_{M_{j-1}+k_j(t-1)+1}\zd d_{M_{j-1}+k_jt}$ are 
distinct, and $\{d_{M_{j-1}+k_j(t-1)+1}\zd d_{M_{j-1}+k_jt}\}=C_{M_{j-1}+k_jt}$.

Since $\ang\Gm_K$ is closed under conjunctions, by the Fleischer 
and Rosenberg result \cite{Fleischer78:Galois} it satisfies  
$\ang\Gm_K=\Inv(\PPol(\ang\Gm_K))$. Moreover, by the proof of 
Theorem~2 of \cite{Fleischer78:Galois} $\rela$ is the set of all 
tuples of the form $f(\vc\bc n)$ for $n\ge1$, $\vc\bc n\in\relo$, 
and $f\in\PPol(\ang\Gm_K)$. Therefore there exist $n\ge1$, 
$\vc\bc n\in\relo$ and $f\in\PPol(\ang\Gm_K)$ such that 
$\bb=f(\vc\bc n)$. Let $\pr_{[\ell]}\bc_q=\ba_{i_q}$. For any 
selection $\vc En$ of $k_j$-element subsets of $D$, $j\in [s]$, there is $t\in [r_{k_j}]$ 
such that $E_q=B^{k_jt}_{i_q}$ for $q\in[n]$. By the choice of $\bb$ 
the range of $f$ on $E_1\tm\ldots\tm E_n=B^{k_jt}_{i_1}\tm\ldots\tm 
B^{k_jt}_{i_n}$ contains $C_{M_{j-1}+k_jt}$. Hence $f$ is $k_j$-subset surjective
for any $k_j\in K$, and so $f\in \mKpPol(\Gm)$, as it is equal to 
$\mKpPol(\ang\Gm_k)$. Therefore $\rel$ is invariant under 
$f$, and so $(b_1\zd b_\ell)\in\rel$. Relation $\rela$ satisfies the 
required conditions, which completes the proof.
\end{proof}

\begin{corollary}\label{cor:Galois}
There is a Galois correspondence between $K$-existential partial
co-clones on one side and partial clones generated by $K$-surjective
partial functions on the other side.

More precisely, for any set $\Gm$ of relations on $D$, any 
$K\sse\{1\zd|D|\}$, and any set $C$ of $K$-surjective partial functions
on $D$, 
\begin{itemize}
\item
$\Inv(C)$ is a $K$-existential partial co-clone;
\item
$\PPol(\Gm)$ is a partial co-clone generated by the set
$\mKpPol(\ang\Gm_K)$ of $K$-surjective partial functions;
\item
$\Inv(\mKpPol(\Gm))=\ang\Gm_K$;
\item
$\mKpPol(\Inv(C))$ is the set of $K$-surjective functions from the 
partial clone generated by $C$.
\end{itemize}
\end{corollary}

\begin{corollary}\label{cor:rel-clones}
Let $\Gm$ be a set of relations on a set $D$. 
\begin{itemize}
\item[(a)]
$\Inv(\mkpPol(\Gm))=\ang\Gm_k$;
\item[(b)]
$\Inv(\mkPol(\Gm))=\dang\Gm_k$;
\item[(c)]
$\Inv(\mPol(\Gm))=\dang\Gm_\infty$;
\end{itemize}
\end{corollary}

\section{The lattice of Boolean max-co-clones}

In this section we give a description of all max-co-clones on $\{0,1\}$. 
We will use the description of usual Boolean co-clones from 
\cite{Post41} and \emph{plain bases} of Boolean co-clones found 
in \cite{Creignou08:plain}. Recall that plain basis of a co-clone 
$C$ is a set $\Gm$ of relations such that the closure of $\Gm$ 
with respect to manipulation of variables and conjunction is $C$. 

To state the results of \cite{Creignou08:plain} and then to proceed with the proof, we need 
some definitions and notation. A relation $\rel(\vc xn)$ is said to be \emph{trivial} if it can 
be specified by giving a set of variables that are equal to 0 (to 1) in every tuple from $\rel$, and 
a collection of conditions of the form $x_i=x_j$. More formally, there are sets $Z,W\sse[n]$ and an 
equivalence relation $\sim$ on $[n]-(Z\cup W)$ such that $\ba\in\rel$ if and only if $\ba[i]=0$ 
whenever $i\in Z$, $\ba[i]=1$ whenever $i\in W$, and $\ba[i]=\ba[j]$ whenever $i\sim j$. 
A relation is called \emph{monotone} if it is invariant with respect to $\join$, the Boolean 
disjunction operation, or $\meet$, the Boolean conjunction operation. Relation $\rel$ is called 
\emph{self-complement} if along with any tuple $\ba\in\rel$ it also contains its \emph{complement},
the tuple $\neg\ba$ such that $\neg\ba[i]=1$ if and only if $\ba[i]=0$. Finally, relation $\rel$ 
is called \emph{affine} if it is the set of solutions to a system of linear equations over $GF(2)$. 
Addition in $GF(2)$ we denote by $\oplus$.

For $I\sse[n]$ we denote by $\ba_I$ the assignment to $\vc xn$ in which $\ba[i]=1$ 
if $i\in I$ and $\ba[i]=0$ otherwise. We will use the following notation: 
$\dl_0,\dl_1$ denote the unary \emph{constant} relations 
$\{(0)\},\{(1)\}$, respectively. $\EQ$ is the binary \emph{equality} relation $\{(0,0),(1,1)\}$;
while $\NEQ$ is the binary \emph{disequality} relation $\{(0,1),(1,0)\}$. $\IMP^k(\vc xk,y)$ 
is the Horn $(k+1)$-ary relation given by the formula $\neg x_1\join\ldots\join\neg x_k\join y$, 
that is, $\ba\in\rel$ if and only if $(\ba[1]\zd\ba[k],\ba[k+1])$ satisfies the formula. By 
$\NIMP^k$ we denote the anti-Horn relation given by the formula $x_1\join\ldots\join x_k\join\neg y$. 
$\OR^k$ denotes the relation $\{0,1\}^k-\{(0\zd0)\}$, and $\NAND^k$ denotes the relation
$\{0,1\}^k-\{(1\zd1)\}$. Finally, $\Compl_{k,\ell}$ is the $(k+\ell)$-ary relation 
$\{0,1\}^{k+\ell}-\{(0\zd0,1\zd1),(1\zd1,0\zd0)\}$, where the first of the two excluded tuples 
contains $k$ zeros and $\ell$ ones, while the second contains $k$ ones and $\ell$ zeros.

Fig.~\ref{fig:post} shows the lattice of Boolean co-clones (borrowed from \cite{Creignou08:plain}), 
and Table~\ref{tab:plain-post} lists plain bases of Boolean co-clones. Table~\ref{tab:plain-post} 
is also taken from \cite{Creignou08:plain} only with notation changed to match the one used
here.

\begin{table}
\begin{tabular}{|l|l|}
\hline 
Co-clone & Plain basis\\
\hline 
$IBF$ &  $\{\EQ \}$ \\
$IR_0$ & $\{\EQ,\dl_0 \}$ \\
$IR_1$ & $\{\EQ,\dl_1 \}$ \\
$IR_2$ &  $\{\EQ,\dl_0,\dl_1 \}$ \\
$IM$ & $\{\IMP\}$ \\
$IM_0$ &  $\{\IMP,\dl_0\}$ \\
$IM_1$ &  $\{\IMP,\dl_1\}$ \\
$IM_2$ &  $\{\IMP,\dl_0,\dl_1\}$ \\
$IS^k_0$ & $\{\EQ\}\cup\{\OR^\ell\mid \ell\le k\}$ \\
$IS_0$ & $\{\EQ\}\cup\{\OR^\ell\mid \ell\in\nat\}$ \\
$IS^k_1$ & $\{\EQ\}\cup\{\NAND^\ell\mid \ell\le k\}$ \\
$IS_1$ & $\{\EQ\}\cup\{\NAND^\ell\mid \ell\in\nat\}$ \\
$IS^k_{02}$ & $\{\EQ,\dl_0\}\cup\{\OR^\ell\mid \ell\le k\}$ \\
$IS_{02}$ &  $\{\EQ,\dl_0\}\cup\{\OR^\ell\mid \ell\in\nat\}$ \\
$IS^k_{12}$ & $\{\EQ,\dl_1\}\cup\{\NAND^\ell\mid \ell\le k\}$ \\
$IS_{12}$ &  $\{\EQ,\dl_1\}\cup\{\NAND^\ell\mid \ell\in\nat\}$ \\
$IS^k_{01}$ & $\{\IMP\}\cup\{\OR^\ell\mid \ell\le k\}$ \\
$IS_{01}$ &  $\{\IMP\}\cup\{\OR^\ell\mid \ell\in\nat\}$ \\
$IS^k_{11}$ &  $\{\IMP\}\cup\{\NAND^\ell\mid \ell\le k\}$ \\
$IS_{11}$ & $\{\IMP\}\cup\{\NAND^\ell\mid \ell\in\nat\}$ \\
$IS^k_{00}$ & $\{\IMP,\dl_0\}\cup\{\OR^\ell\mid \ell\le k\}$ \\
$IS_{00}$ &  $\{\IMP,\dl_0\}\cup\{\OR^\ell\mid \ell\in\nat\}$ \\
$IS^k_{10}$ & $\{\IMP,\dl_1\}\cup\{\NAND^\ell\mid \ell\le k\}$ \\
$IS_{10}$ & $\{\IMP,\dl_1\}\cup\{\NAND^\ell\mid \ell\in\nat\}$ \\
$ID$ & $\{\EQ,\NEQ\}$ \\
$ID_1$ &  $\{\EQ,\NEQ,\dl_0,\dl_1\}$ \\
$ID_2$ &  $\{\dl_0,\dl_1,\OR,\IMP,\NAND\}$ \\
$IL$ & $\{x_1\oplus\ldots\oplus x_k=0\mid k\text{ even}\}$ \\
$IL_0$ & $\{x_1\oplus\ldots\oplus x_k=0\mid k\in\nat\}$ \\
$IL_1$ & $\{x_1\oplus\ldots\oplus x_k=c\mid k\in\nat, k\equiv c\pmod2, c\in\{0,1\}\}$ \\
$IL_2$ &  $\{x_1\oplus\ldots\oplus x_k=c\mid k\in\nat, c\in\{0,1\}\}$ \\
$IL_3$ &  $\{x_1\oplus\ldots\oplus x_k=c\mid k\text{ even}, c\in\{0,1\}\}$ \\
$IV$ &  $\{\IMP^k\mid k\ge1\}$ \\
$IV_0$ & $\{\IMP^k\mid k\ge1\}\cup\{\dl_0\}$ \\
$IV_1$ & $\{\OR^k\mid k\in\nat\}\cup\{\IMP^k\mid k\ge1\}$ \\
$IV_2$ & $\{\OR^k\mid k\in\nat\}\cup\{\IMP^k\mid k\ge1\}\cup\{\dl_0\}$ \\
$IE$ &  $\{\NIMP^k\mid k\ge1\}$ \\
$IE_0$ & $\{\NAND^k\mid k\in\nat\}\cup\{\NIMP^k\mid k\ge1\}$ \\
$IE_1$ &  $\{\NIMP^k\mid k\ge1\}\cup\{\dl_1\}$ \\
$IE_2$ &   $\{\NAND^k\mid k\in\nat\}\cup\{\NIMP^k\mid k\ge1\}\cup\{\dl_1\}$ \\
$IN$ &  $\{\Compl_{k,\ell}\mid k,\ell\ge1\}$ \\
$IN_2$ &  $\{\Compl_{k,\ell}\mid k,\ell\in\nat\}$ \\
$II$ & $\{x_1\join\ldots\join x_k\join\neg y_1\join\ldots\join\neg x_\ell\mid k,\ell\ge1\}$ \\
$II_0$ &  $\{x_1\join\ldots\join x_k\join\neg y_1\join\ldots\join\neg x_\ell\mid k,\ell\ge1\}\cup\{\dl_0\}$ \\
$II_1$ &  $\{x_1\join\ldots\join x_k\join\neg y_1\join\ldots\join\neg x_\ell\mid k,\ell\ge1\}\cup\{\dl_1\}$ \\
$II_2$ &  $\{x_1\join\ldots\join x_k\join\neg y_1\join\ldots\join\neg x_\ell\mid k,\ell\ge1\}\cup\{\dl_0,\dl_1\}$ \\
\hline
\end{tabular}
\caption{Plain bases of Boolean co-clones}
\label{tab:plain-post}
\end{table}

\begin{table}
\begin{tabular}{|l|l|}
\hline 
Max-co-clone & Max-basis\\
\hline 
$IBF$ &  $\{\EQ \}$ \\
$IR_0$ & $\{\EQ,\dl_0 \}$ \\
$IR_1$ & $\{\EQ,\dl_1 \}$ \\
$IR_2$ &  $\{\EQ,\dl_0,\dl_1 \}$ \\
$IM_2$ & $\{\IMP\}$ \\
$IS^k_0$ & $\{\EQ\}\cup\{\OR^k\}$ \\
$IS_0$ & $\{\EQ\}\cup\{\OR^\ell\mid \ell\in\nat\}$ \\
$IS^k_1$ & $\{\EQ\}\cup\{\NAND^k\}$ \\
$IS_1$ & $\{\EQ\}\cup\{\NAND^\ell\mid \ell\in\nat\}$ \\
$IS^k_{02}$ & $\{\EQ,\dl_0,\OR^k\}$ \\
$IS_{02}$ &  $\{\EQ,\dl_0\}\cup\{\OR^\ell\mid \ell\in\nat\}$ \\
$IS^k_{12}$ & $\{\EQ,\dl_1\}\cup\{\NAND^\ell\mid \ell\le k\}$ \\
$IS_{12}$ &  $\{\EQ,\dl_1\}\cup\{\NAND^\ell\mid \ell\in\nat\}$ \\
$ID$ & $\{\EQ,\NEQ\}$ \\
$ID_1$ &  $\{\EQ,\NEQ,\dl_0,\dl_1\}$ \\
$IL$ & $\{x_1\oplus\ldots\oplus x_k=0\mid k\text{ even}\}$ \\
$IL_0$ & $\{x_1\oplus\ldots\oplus x_k=0\mid k\in\nat\}$ \\
$IL_1$ & $\{x_1\oplus\ldots\oplus x_k=c\mid k\in\nat, k\equiv c\pmod2, c\in\{0,1\}\}$ \\
$IL_2$ &  $\{x_1\oplus\ldots\oplus x_k=c\mid k\in\nat, c\in\{0,1\}\}$ \\
$IL_3$ &  $\{x_1\oplus\ldots\oplus x_k=c\mid k\text{ even}, c\in\{0,1\}\}$ \\
$IN_2$ &  $\{\Compl_{3,0}\}$ \\
$II_2$ &  $\{\IMP,\OR\}$ \\
\hline
\end{tabular}
\caption{Max-bases of Boolean max-co-clones}
\label{tab:max-post}
\end{table}

\begin{figure}[ht]
\centerline{\includegraphics{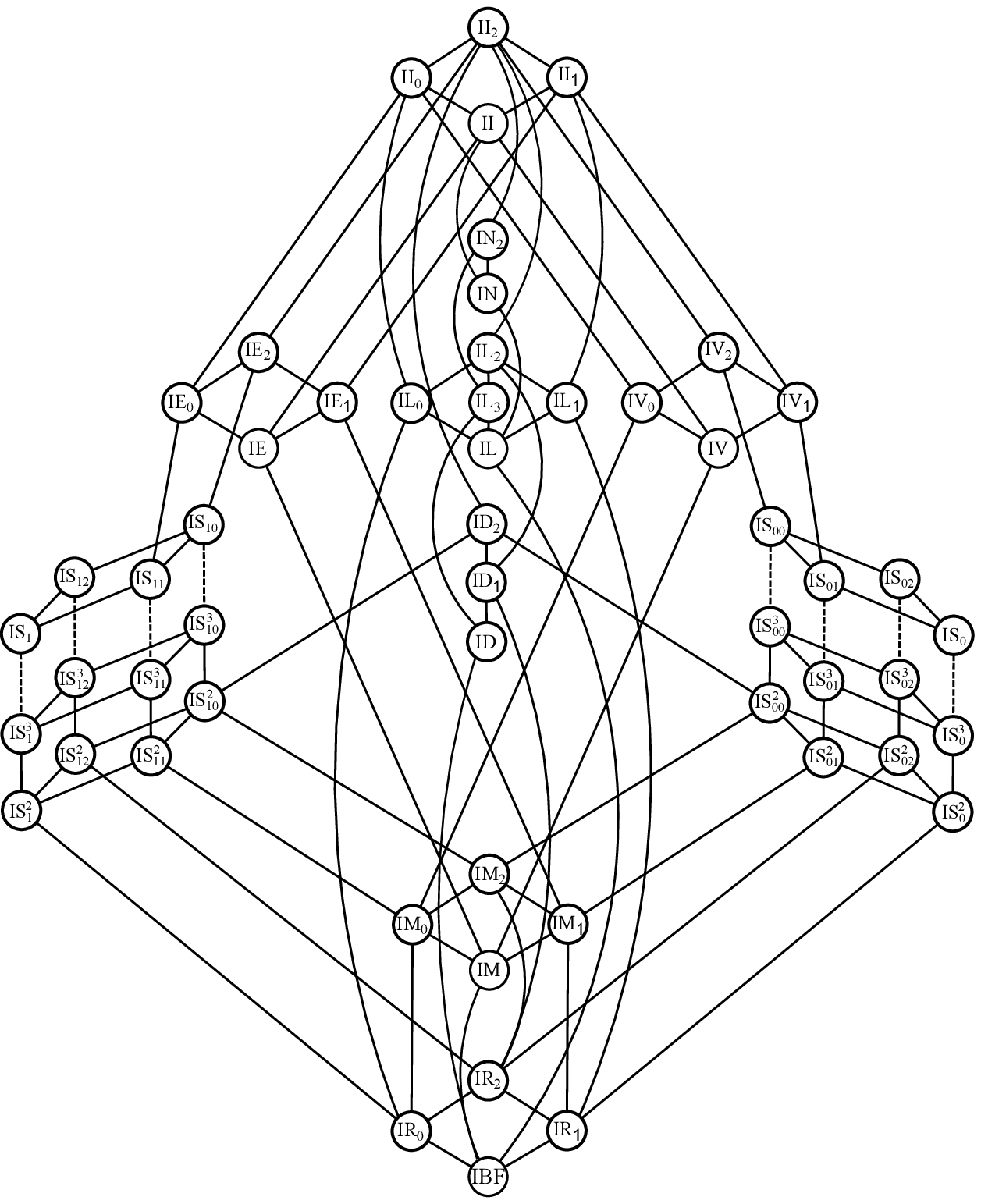}}
\caption{The lattice of Boolean co-clones}
\label{fig:post}
\end{figure}

\begin{figure}[ht]
\centerline{\includegraphics{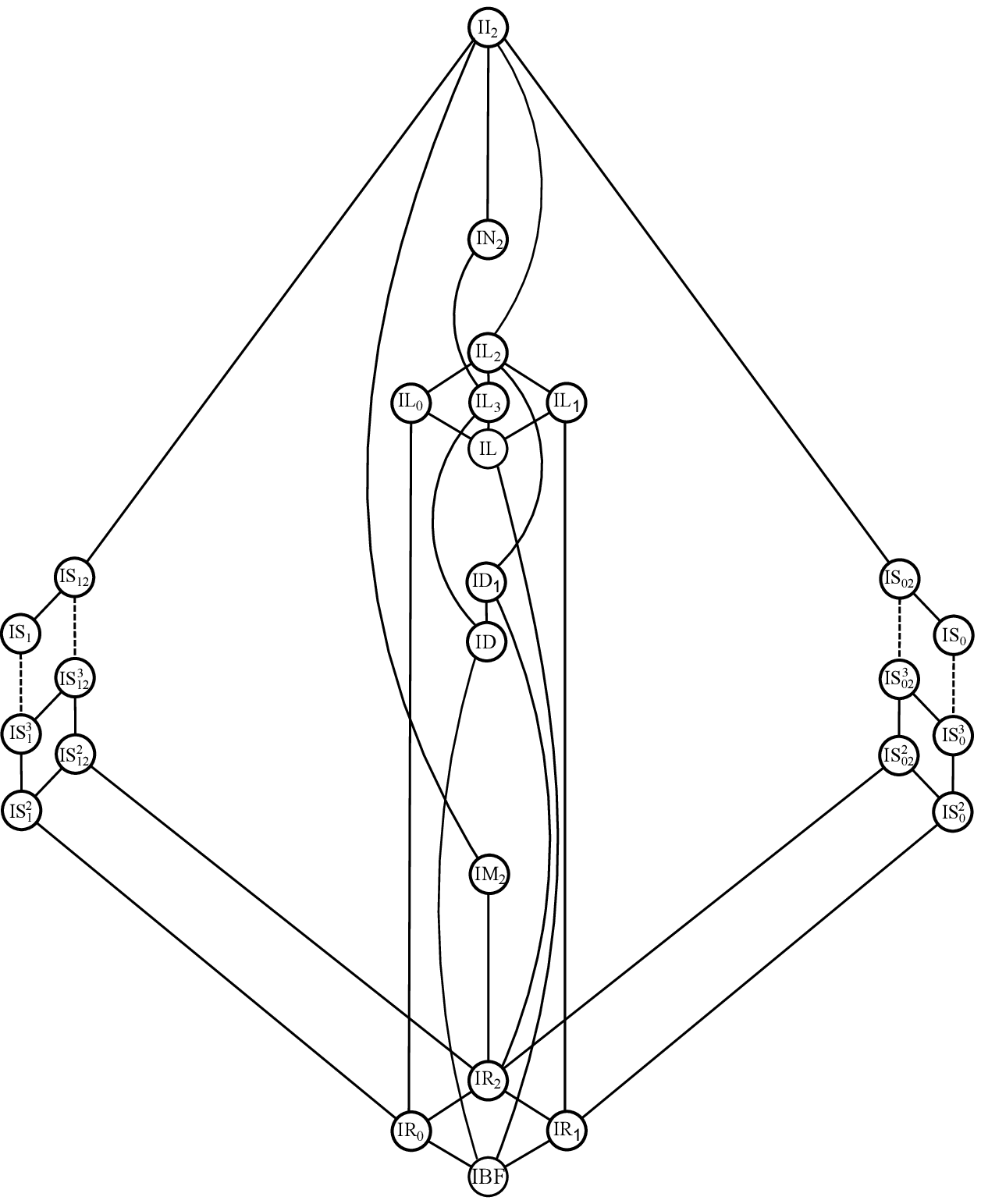}}
\caption{The lattice of Boolean max-co-clones}
\label{fig:max-post}
\end{figure}

The next theorem states the main result of this section.

\begin{theorem}\label{the:max-post}
The lattice of Boolean max-co-clones is shown in Fig~\ref{fig:max-post}. Some generating sets
of these max-co-clones are given in Table~\ref{tab:max-post}.
\end{theorem}

The theorem will follow from a sequence of auxiliary statements. In 
Section~\ref{sec:implementations} we show that using the $\mex$ quantifier we can 
define various relations, and that any relation can be defined by any two nontrivial binary
relations. Then we show, Lemma~\ref{lem:out-max-generation}, that any proper max-co-clone 
must contain only monotone, or only self-complement, or only affine relations. We consider 
these three cases. In the case of affine relations we show that the max-co-clones of such relations 
are exactly regular co-clones, Lemma~\ref{lem:linear-max-clones}. Then we show, 
Proposition~\ref{pro:self-dual}, that there is only one max-co-clone of self-complement 
relations, which contains a non-affine relation, $IN_2$.  Then we show, 
Lemmas~\ref{lem:IM2-clone},\ref{lem:IMP+-generation},  that there is only one 
proper, that is, not $II_2$, the set of all relations, max-co-clone containing $\IMP$, 
and this max-co-clone is $IM_2$. Finally, we consider the four 
remaining infinite chains of co-clones. In Lemma~\ref{lem:IS-property}
we introduce a property that defines them. Then we show, Lemma~\ref{lem:nand-generation}, 
and~\ref{lem:OR-generation} , that there are no other max-co-clones containing $\OR$ 
(for $\NAND$ a dual result holds). Finally, we show that each of these co-clones is 
a max-co-clone.

\subsection{Some implementations}\label{sec:implementations}

We start with several useful observations.

\begin{lemma}\label{lem:simple-implementations}
(1) $\dl_0,\dl_1\in\mang\IMP$;\\
(2) $\dl_0\in\mang{\NEQ,\dl_1}$, $\dl_1\in\mang{\NEQ,\dl_0}$;\\
(3) $\NAND^k\in\mang{\NAND^m}$ for any $k\le m$;\\
(4) $\OR^k\in\mang{\OR^m}$ for any $k\le m$.
\end{lemma}

\begin{proof}
(1) As is easily seen, $\dl_0(x)=\mex y\,\IMP(x,y)$, and $\dl_1(x)=\mex y\, \IMP(y,x)$.

(2) The first inclusion follows from $\dl_0(x)=\mex y(\NEQ(x,y)\meet\dl_1(y))$; the second one is similar.

(3) This claim follows from $\NAND^{m-1}(\vc x{m-1}) = \mex x_m \NAND^m(\vc xm)$.

(4) is similar to (3).
\end{proof}

\begin{lemma}\label{lem:IMP-generation}
For any two different relations $\rel,\rel'\in\{\NEQ,\IMP,\OR,\NAND\}$, 
$\mang{\rel,\rel'}=II_2$, the set of all relations on $\{0,1\}$.
\end{lemma}

\begin{proof}
Observe first that 
\begin{eqnarray*}
\OR\cap\NAND &=& \NEQ,\\
\IMP(x,y) &=& \mex z(\OR(z,y)\meet\NEQ(z,x))\\
 &=& \mex z(\NAND(x,z)\meet\NEQ(z,y))\\
\OR(x,y) &=& \mex z(\IMP(z,y)\meet\NEQ(z,x))\\
 &=& \mex z,t(\NAND(z,t)\meet\NEQ(z,x)\meet\NEQ(t,y))\\
\NAND(x,y) &=& \mex z(\IMP(x,z)\meet\NEQ(z,x))\\
 &=& \mex z,t(\OR(z,t)\meet\NEQ(z,x)\meet\NEQ(t,y)).
\end{eqnarray*}
Also in the relation $\relo(x,y,z,t)=\OR(x,y)\meet\IMP(x,z)\meet\IMP(y,t)$ 
assignments $(0,1)$ and $(1,0)$ to $x,y$ are extendible in two ways, 
while $(1,1)$ is extendible in only one way. Therefore
\begin{eqnarray*}
\NEQ(x,y) &=& \mex(z,t)(\OR(x,y)\meet\IMP(x,z)\meet\IMP(y,t)),\ \ \ \text{and, similarly,}\\
\NEQ(x,y) &=& \mex(z,t)(\NAND(x,y)\meet\IMP(z,x)\meet\IMP(t,y)).
\end{eqnarray*}
Thus $\{\NEQ,\IMP,\OR,\NAND\}\sse\mang{\rel,\rel'}$, and it suffices to show 
that $\mang{\NEQ,\IMP,\OR,\NAND}=II_2$.

The rest of the proof is derived from that of Lemma~15 
\cite{Bulatov11:log-supermodular}, only it does not have to deal with
weights.

Let $\rel(\vc xn)$ be any relation. For each $I\sse[n]$ with 
$\ba_I\in\rel$ introduce a new variable $z_I$. 
Consider the relation given by
$$
\relo=\bigwedge_{I\sse[n],\ba_I\in\rel}\left(\bigwedge_{i\in I}
\IMP(z_I,x_i)\meet\bigwedge_{i\not\in I}\NAND(z_I,x_i)\right).
$$
Every assignment $\ba_I\in\rel$ can be extended to the variables $z_J$ 
in two ways: with $z_I=0$ and $z_I=1$. Any other assignment can be 
extended in only one way. Therefore
$$
\rel(\vc xn)=\mex(z_I)_{I\sse[n],\ba_I\in\rel}\relo,
$$
which completes the proof.
\end{proof}

\begin{lemma}\label{lem:NEQ-generation}
Let $\rel$ be a non-affine relation and $a\in\{0,1\}$. Then $\mang{\rel,\NEQ,\dl_a}=II_2$.
\end{lemma}

\begin{proof}
By Lemma~\ref{lem:IMP-generation} it suffices to prove that one of 
$\IMP,\OR$, or $\NAND$ belongs to $\mang{f,\NEQ,\dl_a}$. 
Observe first that we can always assume that the all-zero tuple 
$\ba_\eps\in\rel$. Indeed, if for some $I\sse[n]$ we have 
$\ba_I\in\rel$ then the relation 
$$
\rel'(\vc xn)=\mex(z_i)_{i\in I}\left(\rel(\vc xn)\meet
\bigwedge_{i\in I}\NEQ(z_i,x_i)\right)
$$
contains $\ba_\eps$. As $\rel\not\in IL_2$, by Lemma~4.10 of \cite{Creignou01:complexity}, there are 
tuples $\ba,\bb,\bc\in\rel$ such that $\bd=\ba\oplus\bb\oplus\bc\not\in\rel$. 
Observing that $\be\in\rel$ if and only if $\be\oplus\ba_I\in\rel'$, 
we have that $\ba\oplus\ba_I,\bb\oplus\ba_I,\bc\oplus\ba_I\in\rel'$, but $\bd\oplus\ba_I=(\ba\oplus\ba_I)\oplus(\bb\oplus\ba_I)\oplus
(\bc\oplus\ba_I)\not\in\rel$. Hence $\rel'$ is not affine as well. 
Also, if $b\in\{0,1\}$ is such that $\{0,1\}=\{a,b\}$ 
then by Lemma~\ref{lem:simple-implementations}(2)  
$\dl_0,\dl_1\in\mang{\rel,\NEQ,\dl_a}$.

Again we use Lemma~4.10 of \cite{Creignou01:complexity} to find 
to find tuples $\ba,\bb,\bc\in\rel$ such 
that $\bd=\ba\oplus\bb\oplus\bc\not\in\rel$. Note that $\ba$ can be chosen
to be the all-zero tuple $\ba_\eps$. 
After rearranging 
variables these tuples can be represented as follows
$$
\begin{array}{r|cccc|l}
\ba&0\ldots0&0\ldots0&0\ldots0&0\ldots0&\in\rel\\
\bb&0\ldots0&0\ldots0&1\ldots1&1\ldots1&\in\rel\\
\bc&0\ldots0&1\ldots1&0\ldots0&1\ldots1&\in\rel\\
\hline
\bd&0\ldots0&1\ldots1&1\ldots1&0\ldots0&\not\in\rel\\
\hline
&x\ldots x&y\ldots y&z\ldots z&t\ldots t
\end{array}
$$
Denote by $\rel'$ the relation obtained from $\rel$ by identifying variables 
as shown in the last row of the table. Relation $\rel'$ contains tuples 
$(0,0,0,0),(0,0,1,1),(0,1,0,1)$ but does not contain $(0,1,1,0)$, and so 
does not belong to $IL_2$. Replacing 
$\rel'$ with
$$
\rel''(x,y,z)=\mex t(\rel(t,x,y,z)\meet\dl_0(t)),
$$
we obtain a relation $\rel''$ such that $(0,0,0),(0,1,1),(1,0,1)\in\rel''$ but $(1,1,0)\not\in\rel''$.

We now proceed depending on which of the 4 remaining tuples 
(a) $(1,0,0)$, (b) $(0,1,0)$, (c) $(0,0,1)$, and (d) $(1,1,1)$ relation 
$\rel''$ contains. If it contains none of (a)--(d) then 
$\NAND(x,y)=\mex z\rel''(x,y,z)$. If it contains (a) or (b) but not 
(d) then $\NAND$ is obtained by identifying $y$ and $z$, or $x$ and $z$, 
respectively. If $\rel''$ contains (c) but not (d) then 
$\NAND(x,y)=\mex z(\rel''(x,y,z)\meet\dl_1(z))$. If it contains (d) 
but not (a) then $\IMP(x,y)=\rel''(x,y,y)$. In the case $\rel''$ contains 
(a), (d), but does not contain (b) $\IMP$ is obtained by identifying 
$x$ and $z$. If $\rel''$ contains (a), (d), and (b) 
$\OR(x,y)=\mex z(\rel''(x,y,z)\meet\dl_1(z))$. Finally, if the relation 
contains all of (a)--(d) $\IMP(y,x)=\rel''(y,y,x)$.
\end{proof}

Next we show that every max-co-clone is a subset of  $IL_2$, $IN_2$, $IV_2$,
or $IE_2$. 

\begin{lemma}\label{lem:out-max-generation}
Let $\Gm$ be a set of relations, which is not affine, monotone, or self-complement. Then $\mang\Gm=II_2$.
\end{lemma}

\begin{proof}
Let $\rel(\vc xn)\in \Gm$ be a non-self-complement relation. Then 
after suitable rearrangement of variables there is $i\in\{0\zd n\}$ 
such that $\ba_{[i]}\in\rel$, while $\ba_{[n]-[i]}\not\in\rel$. 
If $0<i<n$ then identifying variables $\vc xi$ and $x_{i+1}\zd x_n$ 
we obtain a binary relation $\rel'$ that contains $(1,0)$ but does not 
contain $(0,1)$. As is easily seen either $\mex x\rel'$ or $\mex y\rel'$ 
is a constant relation. In the case $i=0$ or $i=n$, identifying all 
variables of $\rel$ we obtain a constant relation. Thus either 
$\dl_0\in\mang\Gm$ or $\dl_1\in\mang\Gm$.

Suppose $\dl_1\in\mang\Gm$. The case $\dl_0\in\mang\Gm$ is similar.
By Lemma~5.30 of \cite{Creignou01:complexity} for any non-affine relation 
$\rel\in\Gm$, the set $\ang{\rel,\dl_1}\sse\mang{\rel,\dl_1}$ contains one of the following relations: 
$\OR,\IMP,\NAND$. If $\NAND\in\mang{\rel,\dl_1}$ then 
$\dl_0(x)=\NAND(x,x)$, and we can make all the arguments below for 
$\dl_0$ and $\NAND$. Therefore we have two cases to consider. Suppose first that 
$\OR\in\mang{\rel,\dl_1}$.
There is a relation $\relo\in\Gm$ that is not invariant under the $\join$ 
operation. Therefore for some tuple $\ba,\bb\in\relo$ the tuple $\ba\join\bb$ 
does not belong to $\relo$. After an appropriate rearrangement of variables these 
tuples can be represented as follows
$$
\begin{array}{r|cccc|l}
\ba&0\ldots0&0\ldots0&1\ldots1&1\ldots1&\in\relo\\
\bb&0\ldots0&1\ldots1&0\ldots0&1\ldots1&\in\relo\\
\hline
\bd&0\ldots0&1\ldots1&1\ldots1&1\ldots1&\not\in\relo\\
\hline
&x\ldots x&y\ldots y&z\ldots z&t\ldots t
\end{array}
$$
Denote by $\relo'$ the relation obtained from $\relo$ by identifying variables 
as shown in the last row of the table. Relation $\relo'$ contains tuples 
$(0,0,1,1),(0,1,0,1)$ but does not contain $(0,1,1,1)$. Then, relation 
$\relo''(x,y,z)=\mex t(\relo'(x,y,z,t)\meet\dl_1(t)\meet\OR(y,z))$ contains
 tuples $(0,0,1),(0,1,0)$ but does not contain $(0,1,1),(0,0,0),(1,0,0)$. 
We have several cases depending on the 3 remaining tuples (a) $(1,1,0)$, 
(b) $(1,0,1)$, (c) $(1,1,1)$. If none of (a)--(c) is in $\relo''$ then 
$\NEQ(x,y)=\mex z\relo''(z,x,y)$. If $\relo''$ contains (a) but not (c) 
(or (b) but not (c)), then $\NEQ(x,y)=\relo''(x,x,y)$ (respectively, 
$\NEQ(x,y)=\relo''(x,y,x)$). If it contains (c) but does not contain (a) 
and (b) then $\IMP(x,y)=\mex z\,\relo''(x,y,z)$.
If $\relo''$ contains both  (b) and (c) then $\IMP(x,y)
=\mex z(\relo''(x,y,z)\meet\dl_1(z))$. Finally if $\relo''$ contains 
(a),(c), but not (b), then $\IMP(x,y)=\mex z(\relo''(y,z,x)\meet\dl_1(z))$.

In either case $\mang\Gm$ contains a constant relation, either $\NEQ$ or 
$\IMP$, and contains one of $\OR,\IMP,\NAND$. If it contains $\NEQ$, we 
are done by Lemma~\ref{lem:IMP-generation}. So suppose $\IMP\in\mang\Gm$. 
Then we also have $\dl_0,\dl_1\in\mang\Gm$. Since $\Gm$ is not monotone, 
as before we can derive relations $\rela_1,\rela_2\in\mang\Gm$ such that 
$(0,0,1,1),(0,1,0,1)\in\rela_1,\rela_2$, but $(0,1,1,1)\not\in\rela_1$, 
$(0,0,0,1)\not\in\rela_2$. Now it is easy to see that 
$\NEQ=\rela'_1\meet\rela'_2$, where $\rela'_i(x,y)
=\mex z\mex t(\rela_i(z,x,y,t)\meet\dl_0(z)\meet\dl_1(t)$.
\end{proof}

\subsection{Affine relations}

Recall that the set of affine relations, that is, ($n$-ary) relations that 
can be represented as the set of solutions to a system of linear equations 
over $\GF(2)$ is denoted by $IL_2$. The next lemma follows from basic linear algebra, as sets of extensions 
of tuples are cosets of the same vector subspace. For the sake of completeness
we give a proof of this lemma.

\begin{lemma}\label{lem:rectangularity}
Let $\rel$ be an ($n$-ary) affine relation. Then for any $I\sse[n]$ any two 
tuples $\ba,\bb\in\pr_I\rel$ have the same number of extensions to tuples 
from $\rel$.
\end{lemma}

\begin{proof}
Let $\rel$ be the set of solutions of a system of linear equations 
$A\cdot\bx=\bc$, where $A$ is a $\ell\tm n$-matrix over $GF(2)$,
$\bx=(\vc xn)^\top$, and $\bc\in\{0,1\}^\ell$. Without loss of 
generality $I=[k]$. Then $A$ can be represented as 
$A=[A_1\mid A_2]$, where $A_1$ is a $\ell\tm k$-matrix and $A_2$ 
is a $\ell\tm(n-k)$-matrix; $\bx$ can be represented as 
$\bx=(\bx^1,\bx^2)^\top$, where $\bx^1=(\vc xk)$, 
$\bx^2=(x_{k+1}\zd x_n)$. Fix $\ba\in\pr_{[k]}\rel$ and set 
$\bc_\ba=\bc\oplus (A_1\cdot\ba)$. The set of extensions of 
$\ba$ is the set of solutions of the system 
$A_2\cdot\bx^2=\bc_\ba$. Clearly, the number of solutions this 
system does not depend on $\ba$, provided the system is consistent.
\end{proof}

\begin{lemma}\label{lem:linear-max-clones}
Let $\Gm\sse IL_2$. Then $\Gm$ is a max-co-clone if and only if it is a co-clone.
\end{lemma}

\begin{proof}
Lemma~\ref{lem:rectangularity} implies that for any ($n$-ary) relation 
$\rel\in IL$ and any set $J=\{\vc ik\}\sse[n]$ the max-implementation 
$\mex (x_{i_1}\zd x_{i_k})$ is equivalent to a sequence of ordinary existential 
quantifiers $\ex x_{i_1}\ldots \ex x_{i_k}$.
\end{proof}

\subsection{Monotone relations}

Recall that a relation is said to be monotone if it is invariant with respect to $\meet$ or 
$\join$. In this section we consider relations invariant under $\join$. 
A proof in the case of relations invariant under $\meet$ is similar. 
A monotone relation is called \emph{nontrivial} if it does not belong to $IR_2$.

\begin{lemma}\label{lem:monotone-generation}
Let $\rel$ be a nontrivial relation invariant under $\join$. Then either 
$\IMP\in\mang\rel$, or $\OR\in\mang\rel$. In particular, if the all-zero 
tuple belongs to $\rel$ then $\IMP\in\mang\rel$.
\end{lemma}

\begin{proof}
 Observe that $\rel$ is not self-complement, because as it follows from 
\cite{Post41} (see also Fig.~\ref{fig:post}) all self complement monotone relations are trivial. Also if 
the all-one tuple does not belong to $\rel$, since $\rel$ is invariant under 
$\join$, some variables of $\rel$ equal 0 in all tuples from $\rel$. 
Such variables can be quantified away, and the resulting relation is nontrivial 
as $\rel$ is nontrivial. We may assume the all-one tuple is in $\rel$. 
 
Suppose first that the all-zero tuple belongs to $\rel$. Therefore there is a 
tuple $\ba\in\rel$ such that its complement does not belong to $\rel$. After 
a suitable rearrangement of variables $\ba=(0\zd0,1\zd1)$. Identify 
variables that take 1 in $\ba$ and also variables that take 0 in $\ba$. 
The resulting relation is $\IMP$.

Suppose now that the all-zero tuple does not belong to $\rel$. Then 
$\dl_1(x)=\rel(x\zd x)$. We also assume that $\rel$ is a nontrivial relation 
of the minimal arity from $\mang\rel$. Let $\vc xn$ be the variables $\rel$ 
depends on. We introduce a partial order on $[n]$ as follows: $i\le_\rel j$ 
iff for any $\ba\in\rel$ $\ba[i]=1$ implies $\ba[j]=1$. If $x_i\le_\rel x_j$ 
for no $i,j\in[n]$, then for any $i\in[n]$ $\rel'=\mex x_i(\rel(\vc xn)\meet\dl_1(x_i))$ 
is a trivial relation, none of its projections equal $\{1\}$, and therefore the 
all-zero tuple belongs to $\rel'$. Hence $\ba_{\{i\}}\in\rel$ where $\ba_{\{i\}}[i]=1$ 
and $\ba_{\{i\}}[j]=0$ for $j\ne i$. Since $\rel$ is invariant under $\join$, 
this implies that $\rel=\OR^n$, and $\OR\in\mang\rel$ by 
Lemma~\ref{lem:simple-implementations}(4).

Next, consider the case when $x_i\le_\rel x_j$ for some $i,j\in[n]$. 
This means there are tuples $\ba,\bb,\bc\in\rel$ such that $\ba[i]=\ba[j]=0$ 
(since the projection of $\rel$ on each variable is $\{0,1\}$), $\bb[i]=0$, 
$\bb[j]=1$ (due to the minimality of $\rel$, there must be a tuple $\bb$ 
with $\bb[i]\ne\bb[j]$), and $\bc$ is the all-one tuple, in particular 
$\bc[i]=\bc[j]=1$. Moreover, as $\rel$ is invariant under $\join$, we may 
assume that $\bb[\ell]=1$ whenever $\ba[\ell]=1$.
After rearranging variables these tuples can be represented as follows
$$
\begin{array}{r|ccc|l}
\ba&0\ldots0&0\ldots0&1\ldots1&\in\rel\\
\bb&0\ldots0&1\ldots1&1\ldots1&\in\rel\\
\bc&1\ldots1&1\ldots1&1\ldots1&\in\rel\\
\hline
&x\ldots x&y\ldots y&z\ldots z
\end{array}
$$
Denote by $\rel'$ the relation obtained from $\rel$ by identifying variables as 
shown in the last row of the table. Relation $\rel'$ contains tuples 
$\ba'=(0,0,1),\bb'=(0,1,1),\bc'=(1,1,1)$. Observe that for no 
$\bd\in\rel'$ we have $\bd[1]=1$ and $\bd[2]=0$. Therefore 
$\IMP(x,y)=\mex u(\rel'(x,y,u)\meet\dl_1(u))$. 
\end{proof}

We first study max-co-clones not containing $\OR$. By 
Lemma~\ref{lem:simple-implementations}(1) and \cite{Creignou08:plain} 
(see also Table~\ref{tab:plain-post}) $\mang\IMP=IM_2$.

\begin{lemma}\label{lem:IM2-clone}
$IM_2$, $IR_2$, $IR_0$, $IR_1$ are max-co-clones.
\end{lemma}

\begin{proof}
Since $IR_2$, $IR_0$, $IR_1$ essentially contain only unary relations, the lemma 
for these co-clones is straightforward.

For $IM_2$ the result actually follows from Lemma~5 of 
\cite{Bulatov11:log-supermodular}. However, as \cite{Bulatov11:log-supermodular} 
uses a different framework, we give a short proof of this result here. Our
proof can be derived from the one from \cite{Bulatov11:log-supermodular}.
Observe first that $\IMP$ satisfies the property of log-supermodularity. 
A function $f\colon\{0,1\}^n\to\mathbb R$ is said to be \emph{log-supermodular} 
if for any $\ba,\bb$
$$
f(\ba)\cdot f(\bb)\le f(\ba\join\bb)\cdot f(\ba\meet\bb).
$$
Here $\meet$ and $\join$ denote componentwise conjunction and disjunction. 
This definition can be extended to relations if they are treated as predicates, 
that is, functions with values $0,1$. As is easily seen, a relation is 
log-supermodular if and only if it is invariant under $\meet$ and $\join$. 
First we show that if $\Gm$ is a set of log-supermodular relations then every 
relation from $\mang\Gm$ is log-supermodular. The property of 
log-supermodularity is obviously preserved by manipulations with variables 
and conjunction, because it is equivalent to the existence of certain polymorphisms. 
Suppose $\rel(\vc xn,\vc ym)$ is log-supermodular and $\relo(\vc xn)
=\mex(\vc ym)\rel(\vc xn,\vc ym)$. We associate every tuple 
$(\ba,\bb)\in\{0,1\}^{n+m}$ with the set of ones in this tuple, and 
therefore can view $\rel$ as a function on the power set of $[n+m]$. 
Take $\ba,\ba'\in\{0,1\}^n$ and prove that
$\relo(\ba)\cdot\relo(\ba')\le\relo(\ba\join\ba')\cdot\relo(\ba\meet\ba')$. 
Let $A$ be the set of tuples of the form $(\ba,\bb)\in\{0,1\}^{n+m}$ and $A'$ the 
set of tuples of the form $(\ba',\bb)\in\{0,1\}^{n+m}$ viewed as subsets of $[n+m]$. 
Also, let $\rel(C)=\sum_{(\bc,\bd)\in C}\rel(\bc,\bd)$ for $C\sse[n+m]$ 
and $f(\vc xn)=\sum_{\vc ym}\rel(\vc xn,\vc yn)$. Denote by $A\join A'$ 
and $A\meet A'$ the sets $A\join A' = \{\bc\join \bc'\mid \bc\in A \text{ and } \bc'\in A'\}$
and $A\meet A' = \{\bc\meet \bc'\mid \bc\in A \text{ and } \bc'\in A'\}$. 
Note that $f(\ba\join \ba')= \rel(A\join A')$ and $f(\ba\meet \ba′) = \rel(A\meet A')$. 
Since $\rel$ is log-supermodular, we know that 
$\rel(\bc,\bd)\cdot\rel(\bc',\bd')\le \rel(\bc\join\bc',\bd\join\bd')\cdot 
\rel(\bc\meet\bc',\bd\meet\bd')$ for all $(\bc,\bd),(\bc',\bd')\in\{0, 1\}^{n+m}$. 
Thus, applying the Ahlswede-Daykin Four-Functions Theorem 
\cite{Alshwede78:inequality} with $\al = \beta =\gm = \dl = \rel$,
\begin{equation}
f(\ba)\cdot f(\ba')=\rel(A)\cdot\rel(A')\le\rel(A\join A')\cdot\rel(A\meet A')
=f(\ba\join\ba')\cdot f(\ba\meet\ba').
\label{equ:lsm}
\end{equation}

Now suppose $\ba,\ba'\in\relo$. This means that $f(\ba)=f(\ba')$ and this 
number is the maximal number of extensions of a tuple from $\{0,1\}^n$ 
to tuples from $\rel$. By (\ref{equ:lsm}) $f(\ba\join\ba'), 
f(\ba\meet\ba')\ne0$ and either $f(\ba\join\ba')\ge f(\ba)$ or 
$f(\ba\meet\ba')\ge f(\ba')$. However, as $f(\ba)$ is the maximal number 
of extensions, strict inequality is impossible, and we get 
$f(\ba\join\ba')= f(\ba\meet\ba')=f(\ba)$. Therefore $(\ba\join\ba'),
(\ba\meet\ba')\in\relo$, and so $\relo(\ba)\cdot\relo(\ba')\le
\relo(\ba\join\ba')\cdot\relo(\ba\meet\ba')$.

Thus $\mang{IM_2}$ contains only log-supermodular relations. However, 
as it was observed above, log-supermodularity of relations is equivalent to 
invariance under $\meet$ and $\join$. Since, $IM_2$ is the class of all 
relations invariant under this two operations, we have $\mang{IM_2}=IM_2$. 
\end{proof}

\begin{lemma}\label{lem:IMP+-generation}
Let $\rel\not\in IM_2$. Then $\mang{\rel,\IMP}=II_2$.
\end{lemma}

\begin{proof}
If $\rel$ is not invariant under $\join$ and $\meet$ then the result follows by 
Lemma~\ref{lem:out-max-generation}, since $\IMP$ is not affine or 
self-complement. Suppose $\rel$ is invariant with respect $\join$. 

Recall that a relation $\relo(\vc xn)$ is called \emph{2-decomposable} if any 
tuple $\ba$ such that $(\ba[i],\ba[j])\in\pr_{\{i,j\}}\relo$ for all $i,j\in[n]$ belongs to $\relo$. 

\smallskip

{\sc Case 1.}
$\rel$ is not 2-decomposable.

\smallskip

Let $I\sse[n]$ be a minimal set such that $\pr_I\rel$ is not 2-decomposable, 
clearly, $|I|\ge3$. Let $\rel'=\pr_I\rel$. There is $\ba\in\{0,1\}^{|I|}$ such 
that for any $i\in I$ $\ba_i\in\rel'$, where $\ba_i$ denotes the tuple such that 
$\ba_i[i]\ne\ba[i]$ and $\ba_i[j]=\ba[j]$ for $i\ne j$. Choose $i_1,i_2,i_3\in I$, and set 
$I-\{i_1,i_2,i_3\}=\{i_4\zd i_k\}$ and 
$$
\relo=\mex x_{i_4}\ldots \mex x_{i_k} (\rel(\vc xn)\meet\dl_{\ba[i_4]}(x_{i_4})
\meet\ldots\meet\dl_{\ba[i_k]}(x_{i_k})).
$$
As is easily seen, $\relo$ is not 2-decomposable, and moreover, 
$\pr_{\{i_1,i_2,i_3\}}\relo$ is not 2-decomposable. Let 
$\relo'=\pr_{\{i_1,i_2,i_3\}}\relo$. There is $\ba\in\{0,1\}^3$ such that 
for any $i\in I$ $\ba_i\in\relo'$, where $\ba_i$ denotes the tuple such 
that $\ba_i[i]\ne\ba[i]$ and $\ba_i[j]=\ba[j]$ for $i\ne j$. Observe that there are at 
most one 1 among components of $\ba$. Indeed, if, say, $\ba=(1,1,0)$ 
then $\ba=\ba_1\join\ba_2\in\relo'$. Suppose first that $\ba$ is the 
all-zero tuple. Then after rearranging variables these tuples can be 
represented as follows
$$
\begin{array}{r|ccccccccccc|l}
\arraycolsep-3pt
\ba_1&1&0&0&0\ldots0&0\ldots0&0\ldots0&1\ldots1&0\ldots0&1\ldots1&1\ldots1&1\ldots1&\in\rel\\
\ba_2&0&1&0&0\ldots0&0\ldots0&1\ldots1&0\ldots0&1\ldots1&0\ldots0&1\ldots1&1\ldots1&\in\rel\\
\ba_3&0&0&1&0\ldots0&1\ldots1&0\ldots0&0\ldots0&1\ldots1&1\ldots1&0\ldots0&1\ldots1&\in\rel\\
\hline
\ba&0&0&0&*&*&*&*&*&*&*&*&\not\in\rel\\
\hline
&x&y&z&t_1&t_2&t_3&t_4&t_5&t_6&t_7&t_8&
\end{array}
$$
Denote by $\relo''$ the relation obtained from $\relo$ by identifying variables as 
shown in the last row of the table. Then set 
$$
\rela(x,y,z,t,u,v)=\mex t_1\mex t_8 (\relo''(x,y,z,t_1,z,y,x,t,u,v,t_8)\meet
\dl_0(t_1)\meet\dl_1(t_8)).
$$
Relation $\rela$ contains tuples $\bb_1=(1,0,0,1,1,0),\bb_2=(0,1,0,1,0,1), 
\bb_3=(0,0,1,0,1,1)$ but does not contain $(0,0,0,a,b,c)$ for any $a,b,c\in\{0,1\}$. 
Next we set $\rela'(x,y,z)=\mex t,u,v (\rela(x,y,z,t,u,v)\meet\dl_1(t)\meet
\dl_1(u)\meet\dl_1(v))$. Since $\rela$ is invariant under $\join$, it contains 
$\bb_1\join\bb_2,\bb_2\join\bb_3,\bb_3\join\bb_1$, and therefore $\rela'$ 
contains tuples $(1,1,0),(1,0,1),(0,1,1),(1,1,1)$, but does not contain $(0,0,0)$. 
Let also $\rela''(x,y,z)=\rela'(x,y,z)\meet\rela'(z,x,y)\meet\rela'(y,z,x)$. As is 
easily seen $\rela''$ is either $\OR^3$ or $\{(1,1,0),(1,0,1),(0,1,1),(1,1,1)\}$. 
In the former case we are done, while in the latter case we just observe that 
$\OR(x,y)=\mex z(\rela''(x,y,z)\meet\dl_1(z))$.

Now suppose $\ba=(0,0,1)$. As before we can construct a relation $\rela$ such 
that $\bb_1=(0,0,0,1,1,1),\bb_2=(0,1,1,0,0,1),\bb_3=(1,0,1,0,1,0)$ belong to 
$\rela$, but  $(0,0,1,a,b,c)$ does not belong to $\rela$ for any $a,b,c\in\{0,1\}$. 
Since $\rel$ is invariant under $\join$ tuples 
$\bb_2\join\bb_1,\bb_3\join\bb_1,\bb_2\join\bb_3\join\bb_1$ also belong to 
$\rela$. Hence $(0,0,0,1),(0,1,1,1),$\lb$(1,0,1,1),(1,1,1,1)\in\rela'(x,y,z,t)
=\rela(x,y,z,t,t,t)$, and $(0,0,1,1)\not\in\rela'$. Therefore\lb
$\OR(x,y)=\mex z\mex t (\rela'(x,y,z,t)\meet\dl_1(z)\meet\dl_1(t))$.

\smallskip

{\sc Case 2.}
$\rel$ is 2-decomposable.

\smallskip

Since $\mang\IMP$ contains $IM_2$ and therefore all 2-decomposable relations 
whose binary projections are either trivial relations or $\IMP$, relation $\rel$ has 
to have a binary projection which is not one of them. As it and all its projections 
are invariant under $\join$, the only nontrivial binary projections it may have are 
$\IMP$ and $\OR$. Therefore for some $i,j\in[n]$ $\pr_{\{i,j\}}\rel=\OR$. 
There are $\ba,\bb,\bc\in\rel$ such that $\ba[i]=\bb[j]=0$ and 
$\ba[j]=\bb[i]=\bc[i]=\bc[j]=1$, but for no $\bd\in\rel$ $\bd[i]=\bd[j]=0$. 
Note also that $\bc$ can be replaced with $\bc\join\ba\join\bb$.
After rearranging variables these tuples can be represented as follows
$$
\begin{array}{r|ccccccc|l}
\ba&0&1&0\ldots0&0\ldots0&0\ldots0&1\ldots1&1\ldots1&\in\rel\\
\bb&1&0&0\ldots0&0\ldots0&1\ldots1&0\ldots0&1\ldots1&\in\rel\\
\bc&1&1&0\ldots0&1\ldots1&1\ldots1&1\ldots1&1\ldots1&\in\rel\\
\hline
\bd&0&0&*&*&*&*&*&\not\in\rel\\
\hline
&x&y&z_1\ldots z_1&z_2\ldots z_2&z_3\ldots z_3&z_4\ldots z_4&z_5\ldots z_5&
\end{array}
$$
Denote by $\rel'$ the relation obtained from $\rel$ by identifying variables as 
shown in the last row of the table. Then set 
$$
\relo(x,y,z)=\mex z_1\mex z_5 (\relo(x,y,z_1,z,x,y,z_5)\meet\dl_0(z_1)\meet\dl_1(z_5)).
$$
Relation $\relo$ contains tuples $(0,1,0),(1,0,0),(1,1,1)$, and $(1,1,0)$, as it is 
invariant under $\join$, but does not contain $(0,0,a)$ for any $a\in\{0,1\}$. 
Then $\OR(x,y)=\mex z(\relo(x,y,z)\meet\dl_0(z))$.
\end{proof}

Next we consider max-co-clones containing $\OR$, but not $\IMP$.

Let $\rel(\vc xn)$ be a relation. If $i,j\in[n]$ are such that $\ba[i]=\ba[j]$ for 
any $\ba\in\rel$, we write $i\sim_\rel j$. Clearly, $\sim_\rel$ is an equivalence 
relation on $[n]$; its class containing $i$ will be denoted by $S_\rel(i)$ or 
$S_\rel(x_i)$. Let also $O_\rel$ denote the set of variables $x_j$ such that 
there is $\bb\in\rel$ with $\bb[j]=1$. An $n$-tuple $\ba$ is said to be 
\emph{$\sim_\rel$-conforming} if (a) $\ba[i]=\ba[j]$ whenever $i\sim_\rel j$, 
and (b) $\ba[i]=0$ whenever $i\not\in O_\rel$. When considered ordered with 
respect to the natural component-wise order ($0\le 1$), $\sim_\rel$-conforming 
tuples form a poset isomorphic to $\{0,1\}^{k_\rel}$, where $k_\rel$ is the 
number of $\sim_\rel$-classes except for the class $[n]-O_\rel$. In what follows 
$\le$ and $<$ will denote relations on the set of $\sim_\rel$-conforming tuples 
for appropriate $\rel$. We say that a relation $\rel(\vc xn)$ satisfies the 
\emph{filter property} if for any $\ba\in\rel$ any $\sim_\rel$-conforming tuple $\ba'$ 
with $\ba\le\ba'$ belongs to $\rel$. The filter property implies that if $\rel$ is 
considered as a subset of the ordered set $\{0,1\}^{k_\rel}$, then it is an 
order filter in this set. In particular, it is completely determined by its minimal 
(with respect to $\le$) elements, or equivalently by the maximal elements not 
belonging to $\rel$. We say that  $\rel$ satisfies the \emph{$r$-filter property}, 
if it satisfies the filter property, and every maximal  tuple  not belonging to $\rel$ 
contains zeros in at most $r$ classes of $\sim_\rel$ from $O_\rel$.

\begin{lemma}\label{lem:IS-property}
(1) A relation $\rel$ belongs to $IS_{12}$ if and only if it satisfies the filter property.\\
(2) A relation $\rel$ belongs to $IS^r_{12}$ if and only if it satisfies the $r$-filter property.
\end{lemma}

\begin{proof}
(1) Suppose $\rel(\vc xn)\in IS_{12}$. Then by Proposition~3 of \cite{Creignou08:plain} 
the set $\EQ,\dl_0,\dl_1$ and $\OR^m$, $m\ge2$ is a plain basis of $IS_{12}$, and 
therefore $\rel$ can be represented by a conjunctive formula $\Phi$ containing variables $\vc xn$, 
relations $\EQ,\dl_0,\dl_1$, and $\OR^m$. Let $\ba\in\rel$, and let $\bb$ be a 
$\sim_\rel$-conforming tuple such that $\ba\le\bb$. We show that it belongs to $\rel$. 
Clearly, $\bb$ satisfies all the $\dl_1$ relations. Also, it satisfies all the $\dl_0$ 
relations, if $\dl_0(x_j)$ belongs to $\Phi$ then 
$j\not\in O_\rel$ and $\bb[j]=0$. Since $\bb$ contains 0 only in the positions $\ba$ 
does, every relation $\OR^m$ is satisfied by $\bb$. Finally, if $\EQ(x_{j_1},x_{j_2})$ 
belongs to $\Phi$, then $j_1\sim_\rel j_2$, therefore all the $\EQ$ relations remain 
satisfied by $\bb$.

Suppose now that $\rel(\vc xn)$ satisfies the filter property. Let $W,Z\sse[n]$ be the sets 
of variables such that for all $\ba\in\rel$ $\ba[i]=1$ (respectively, $\ba[i]=0$) for $i\in W$ 
($i\in Z$). Let also $\vc \ba\ell$ be the maximal tuples not from $\rel$. By $Z_j$ we 
denote the set of $i\in O_\rel$ such that $\ba_j[i]=0$. Suppose $Z_j$ contains elements from $m_j$ classes of $\sim_\rel$. We construct a formula $\Phi$ 
using variables $\vc xn$ and relations $\EQ,\dl_0,\dl_1,\OR^m$, and prove that 
it represents $\rel$. Formula $\Phi$ includes\\
(1) $\dl_0(x_i)$ for each $i\in Z$ and $\dl_1(x_i)$ for each $i\in W$;\\
(2) $\EQ(x_i,x_j)$ for any pair $x_i,x_j$, $i\sim_\rel j$;\\
(3) $\OR^{m_j}(x_{i_1}\zd x_{i_{m_j}})$ for any $\ba_j$, $j\in[\ell]$, and any $\vc i{m_j}$ such 
that $\vc i{m_j}$ belong to different $\sim_\rel$-classes from $Z_j$.\\
Let the resulting relation be denoted by $\relo$. By what is proved above $\relo$ 
satisfies the filter property.  It is straightforward that $O_\relo=O_\rel$ and the 
maximal tuples not in $\relo$ are the same as those of $\rel$.  Therefore $\relo=\rel$.

(2) Suppose first that $\rel$ satisfies the $r$-filter property. Then it can be represented 
by a formula $\Phi$ as in part (1) and for every relation $\OR^m$ used $m\le r$. 
Therefore $\rel\in IS^r_{12}$.

Let now $\rel(\vc xn)\in IS^r_{12}$, and therefore can be represented by a formula 
$\Phi$ in $\vc xn$, and relations $\EQ,\dl_0,\dl_1$, and $\OR^m$ for $m\le r$. We 
need to study the structure of maximal tuples from the complement of $\rel$. We 
use the notation from part (1). Let $\ba$ be such a tuple. It is $\sim_\rel$-conforming, 
so, $\ba[i]=0$ for all $i\in Z$, and $\ba[i]=\ba[j]$ for any $i\sim_\rel j$. This means 
that $\ba$ satisfies all the $\dl_0$ and $\EQ$ relations in $\Phi$. If $\ba$ violates a 
relation $\dl_1$ and there is $i\not\in W$ such that $\ba[i]=0$ then $\ba$ is not 
maximal in the complement of $\rel$. Therefore $\ba[i]=0$ if and only if $i\in W$, 
and $W$ is a single $\sim_\rel$-class.  Suppose $\ba$ violates a relation 
$\OR^m(x_{i_1}\zd x_{i_m})$, and let $D=S(i_1)\cup\ldots\cup S(i_m)$. If 
there is $i\in O_\rel-D$ such that $\ba[i]=0$ then the tuple $\bb$ given by 
$\bb[j]=1$ if $j\in S(i)$ and $\bb[j]=\ba[j]$ otherwise does not belong to 
$\rel$ and $\ba<\bb$, a contradiction. Therefore the set of zeros of any 
maximal tuple from the complement of $\rel$ spans at most $r$ 
classes of $\sim_\rel$, as required.
\end{proof}

Let $\Gm$ be a max-co-clone of monotone relations. By $\orr(\Gm)$ we denote the maximal 
$m$ such that $\OR^m\in\mang\Gm$.
If a maximal number $m$ does not exist we set $\orr(\Gm)=\infty$. 

\begin{lemma}\label{lem:nand-generation}
For any set $\Gm\sse IS_{12}$ of monotone relations 
$$
\mang\Gm=\mang{\{\OR^m\mid m\le\orr(\Gm)\}}\quad \text{or}\quad
\mang\Gm=\mang{\{\OR^m\mid m\le\orr(\Gm)\}}\cup\{\dl_0\}.
$$
\end{lemma}

\begin{proof}
It suffices to show that if $\Gm$ contains a relation $\rel$ with a maximal 
tuple that spans $k$ classes of $\sim_\rel$, then $\OR^k\in\mang\Gm$. Let $\rel$ be such a relation. 
Applying $\mex$ we may assume that the sets $W$ and $Z$ for $\rel$ are 
empty; applying identification of variables we may assume that every set 
$S(i)$ is a singleton. Now let $\ba$ be a maximal tuple that spans $k$ classes of $\sim_\rel$, and 
$I$ the set of positions such that $\ba[i]=0$ if and only if $i\in I$; without 
loss of generality assume $I=[k]$.  Since $\rel$ satisfies the filter property, 
for any $(\vc bk)\in\pr_{[k]}\rel$ the tuple $(\vc bk,1\zd1)$ belongs to 
$\rel$. Observe that identifying all the variables of $\rel$ we make sure that
$\dl_1\in\mang\Gm$. Therefore the relation given by
$$
\relo(\vc xk)=\mex(x_{k+1}\zd x_n)(\rel(\vc xn)\meet\dl_1(x_{k+1})\meet\ldots\meet\dl_1(x_n))
$$
belongs to $\mang\Gm$. It remains to show that $\relo=\OR^k$. By the filter property of $\rel$ 
for any $\vc bk$ that are not all zeros $(\vc bk,1\zd1)\in\rel$. Therefore $(\vc bk)\in\relo$. 
On the other hand, $(0\zd0,1\zd1)\not\in\rel$.
 
It remains to show that for any $\rel(\vc xn)\in IS_{12}$ such that $\ba_{[n]}\not\in\rel$ 
(the all-ones tuple), 
$\dl_0\in\mang\rel$. By the filter property of $\rel$ if $\ba_{[n]}\not\in\rel$ there is 
$i\in[n]$ such that $\ba[i]=0$ for all $\ba\in\rel$. Let $I\sse[n]$ be the set of all 
such coordinate positions; without loss of generality we may assume that $I=[m]$. 
Since $\dl_1\in\mang\rel$, we have
$$
\dl_0(x)=\mex y (\rel(x\zd x,y\zd y)\meet\dl_1(y)),
$$
where $x$ is in the first $m$ positions.
\end{proof}

\begin{lemma}\label{lem:IS-clones}
Every co-clone $IS_1,IS_{12},IS^r_1,IS^r_{12}$ for $r\in\{2,3,\ldots\}$ 
is a max-co-clone.
\end{lemma}

\begin{proof}
First we show that every $IS_{12}, IS^r_{12}$ is a max-co-clone. By 
Lemma~\ref{lem:IS-property} it suffices to prove that if every relation from 
$\Gm$ satisfies the filter or $r$-filter property, then so does every relation from 
$\mang\Gm$. These properties are preserved by manipulations with variables 
and conjunction, because $IS_{12}, IS^r_{12}$ are co-clones. It remains to 
show that they are also preserved by max-implementation. 

Suppose $\rel(\vc xn,\vc ym)$ satisfies the filter property and 
$\relo(\vc xn)=\mex(\vc ym)$\lb$\rel(\vc xn, \vc ym)$. Observe that we may assume 
that for any $x_i$ the set $S(x_i)$ does not contain any variable $y_j$. Indeed, 
if $\ba[i]=\bb[j]$ for any assignment $(\ba,\bb)$ that satisfies $\rel$, then we 
can identify these two variables, and denote the new variable by $x_i$. The number 
of extensions of any assignment to $\vc xn$ does not change, therefore the relation 
$\relo$ defined in the same way from the new relation does not change. 

Choose 
a representation $\Phi$ of $\relo$ that uses $\OR^r$, $\EQ$, $\dl_0,\dl_1$. Such 
a representation exists as the listed relations constitute a plain basis for $IS_{12}$ 
by \cite{Creignou08:plain} (see Table~\ref{tab:plain-post}). Take $\ba\in\relo$ and $x_i\in O_\relo$; let $\ba'$ be 
the tuple such that $\ba\le\ba'$. 
It suffices to verify that every extension $\bb$ of $\ba$ is also extension of $\ba'$. 
Indeed, if this is the case, since $\ba$ has the maximum number of extensions, so does 
$\ba'$, and thus $\ba'\in\relo$. Suppose $(\ba,\bb)\in\rel$. Then $(\ba',\bb)$ 
satisfies every relation $\OR^r$ from $\Phi$, as this tuple contains 1 in every 
position $(\ba,\bb)$ does. It also satisfies every relation $\EQ$, because there is 
no relation of the form $\EQ(x_\ell,y_j)$, and $\ba'[i]=\ba'[j]$ whenever $i\sim_\rel j$. 
Finally, $\dl_0$ and $\dl_1$ are also satisfied, because no 
value is changed in the scopes of the former, and no value is changed to 0 in the 
scope of the latter.

Next we prove that the number of $\sim_\rel$-classes spanned by zeros
of maximal tuples from the complement of $\relo$ does not exceed that of $\rel$. More 
precisely we show that (1) $S_\rel(x_i)\cap\{\vc xn\}\sse S_\relo(x_i)$ for any 
$i\in[n]$, and (2) for every maximal tuple $\ba\not\in\relo$ there is 
$\bb\in\{0,1\}^m$ such that $(\ba,\bb)$ is a maximal tuple not belonging to $\rel$. 

The first claim is obvious, as $\relo\sse\pr_{[n]}\rel$ and therefore if 
$\ba[i]=\ba[j]$ for any $(\ba,\bb)\in\rel$ then $\bc[i]=\bc[j]$ for any $\bc\in\relo$. 
Observe that we may assume that $\pr_j\rel=\{0,1\}$ for any $\j\in\{n+1\zd n+m\}$, 
since otherwise such a variable does not affect the number of extensions of 
tuples from $\pr_{[n]}\rel$. For the second claim let $\ba$ be a maximal tuple 
not belonging to $\relo$. Suppose 
first that $\ba\not\in\pr_{[n]}\rel$. Since for any $\ba'\in\pr_{[n]}\rel$ the tuple 
$(\ba',1\zd1)$ belongs to $\rel$, the tuple $(\ba,1\zd1)$ is a maximal tuple not 
belonging to $\rel$. Next assume $\ba\in\pr_{[n]}\rel$. Let $E(\bc)$ denote the 
set of extensions of a tuple $\bc\in\pr_{[n]}\rel$ to a tuple from $\rel$. 
Due to the filter property of $\rel$ and the assumption that no set $S(x_i)$ contains 
any $y_j$, if $\bc\le\bc'$ then $E(\bc)\sse E(\bc')$. As $\ba$ is a maximal tuple not 
belonging to $\relo$, the number of extensions of any tuple $\ba'$, $\ba<\ba'$, 
is the same, including the all-one tuple $\ba_{[n]}$. However, for any such tuple $\ba'$, 
$E(\ba')\sse E(\ba_{[n]})$ and yet $|E(\ba')|=|E(\ba_{[n]})|$ implying $E(\ba')=E(\ba_{[n]})$. 
Since $|E(\ba)|<|E(\ba')|$ for any tuple $\ba'$, $\ba<\ba'$, there is $\bb$ such 
that $(\ba,\bb)\not\in\rel$ and $(\ba',\bb)\in\rel$ for any tuple $\ba'$, $\ba<\ba'$. 
Choose a maximal $\bb'$, $\bb\le\bb'$, with this property. We need to show that 
$(\ba,\bb')$ is a maximal tuple not belonging to $\rel$. For any $\bb''>\bb'$ 
the tuple $(\ba,\bb'')\in\rel$, because, by the choice of $\bb'$, it is a maximal 
tuple such that $(\ba,\bb')\not\in\rel$. For any $\ba'$, $\ba<\ba'$, the tuple 
$(\ba',\bb)$ belongs to $\rel$, and therefore $(\ba',\bb')\in\rel$. 

Next we show that $\mang{IS^r_1}=IS^r_1$. Co-clone $IS^r_1$ contains
all relations from $IS^r_{12}$ invariant under the constant function~1. So, we prove that any 
relation $\rel\in\mang{IS_1}$ contains the all-one tuple. Relations $\EQ,\dl_1$, 
and $\OR^r$ satisfy this condition. Manipulations with variables and conjunction 
preserves this property. It remains to verify that $\mex$ also preserves this 
property in $IS_{12}$. Let $\rel(\vc xn,\vc ym)\in IS_{12}$ and 
$(1\zd1,1\zd1)\in\rel$. Let also $\relo(\vc xn)=\mex(\vc ym)\rel(\vc xn,\vc ym)$. 
As before we may assume that for any $x_i$ the set $S(x_i)$ does not contain 
any variable $y_j$. Then since $E(\ba)\sse E(\ba_{[n]})$, where $\ba_{[n]}$ is the all-one 
tuple, for any $\ba\in\pr_{[n]}\rel$, $\ba_{[n]}\in\relo$. 
\end{proof}

\begin{lemma}\label{lem:OR-generation}
Let $\rel\not\in IS_{12}$, then $\mang{\rel,\OR}=II_2$.
\end{lemma}

\begin{proof}
First of all $\rel$ can be assumed to be closed under $\join$. Indeed, $\OR$ is not 
self-complement, affine, or closed under $\meet$; so if $\rel$ is not closed under $\vee$
the result follows from Lemma~\ref{lem:out-max-generation}. We also may assume 
that every unary projection of $\rel$ contains two elements. Next, 
observe that we can also assume that for each variable $x$ of $\rel$ the set $S(x)$ 
contains only one element. Indeed, construct a relation $\rel'$ by identifying all 
variables in every set of the form $S(x)$. It now suffices to verify that 
$\rel'\not\in IS_{12}$ whenever $\rel\not\in IS_{12}$. To see this note that $\rel$ 
can be obtained from $\rel'$ through adding new variables and imposing equality 
relations.

If $\rel$ contains the all-zero tuple then by Lemma~\ref{lem:monotone-generation} 
$\IMP\in\mang\rel$ and the result follows from Lemma~\ref{lem:IMP-generation}.

Suppose that the all-zero tuple does not belong to $\rel$. We show that either $\rel$ 
satisfies the filter property, and therefore belongs to $IS_{12}$, or there is a nontrivial 
relation $\relo\in\mang\rel$ containing the all-zero tuple. By what is proved above it 
implies the result. 

For $\ba\in\rel$ we denote by $\rel_\ba$ the relation obtained as follows. Let $O(\ba)$ 
denote the set of coordinate positions in which $\ba$ equals 1. Then
$$
\rel_\ba=\mex(x_i)_{i\in O(\ba)}(\rel(\vc xn\meet\bigwedge_{i\in O(\ba)}\dl_1(x_i)).
$$
If $\rel_\ba$ is a nontrivial relation then we are done, since the all-zero tuple belongs 
to $\rel_\ba$. Therefore assume that every relation $\rel_\ba$ is trivial. Observe that 
since $\ba\join\bb\in\rel$ for any $\bb\in\rel$ and $\pr_{[n]-O(\ba)}(\ba\join\bb)=
\pr_{[n]-O(\ba)}\bb$, we have $\rel_\ba=\pr_{[n]-O(\ba)}\rel$. Therefore every set of the 
form $S(x)$ for $\rel_\ba$ is 1-element. Hence $\rel_\ba=\{0,1\}^{n-|O(\ba)|}$. 
In particular, for any $\ba\in\rel$ and any $i\not\in O(\ba)$ the tuple $\bb$ obtained 
from $\ba$ by changing $\ba[i]$ to 1 belongs to $\rel$. Thus $\rel$ satisfies the filter 
property.
\end{proof}

\begin{prop}\label{pro:infinite-chains}
Every max-co-clone of monotone relations containing a nontrivial relation equals one of 
$IS_1$, $IS_{12}$, $IS^i_1$, $IS^i_{12}$ for $i\in\{2,3,\ldots\}$, $IM_2$. 
\end{prop}

\begin{proof}
By Lemmas~\ref{lem:IM2-clone} and~\ref{lem:IS-clones} all these sets are 
max-co-clones. By Lemma~\ref{lem:IMP+-generation} and the observation that 
$\mang\IMP=IM_2$, max-co-clone $IM_2$ is the only max-co-clone containing 
$\IMP$. By Lemma~\ref{lem:OR-generation} $IS_{12}$ is the greatest 
max-co-clone containing $\OR$. Thus it remains to prove that there are no 
max-co-clones containing $\OR$ and different from 
$IS_1,IS_{12},IS^i_1,IS^i_{12}$ for $i\in\{2,3,\ldots\}$. It follows from 
Lemma~\ref{lem:nand-generation}.
\end{proof}

\subsection{Self-complement max-co-clones}

In this section we consider the remaining case of self-complement max-co-clones.

\begin{prop}\label{pro:self-dual}
There is only one max-co-clone of self-complement relations that is not a subclone of $IL_2$. 
It is $IN_2$, the clone of all self-complement relations.
\end{prop}

The proposition follows from  the following four lemmas.

\begin{lemma}\label{lem:IN-clone}
$IN_2$ is a max-co-clone.
\end{lemma}

\begin{proof}
We need to prove that $IN_2$ is closed under manipulations with variables, 
conjunction, and max-implementation. Since $IN_2$ is a co-clone, it is closed 
under the first two operations. Let $\rel(\vc xn,\vc ym)\in IN_2$ and 
$\relo(\vc xn)=\mex(\vc ym)\rel(\vc xn,\vc ym)$. Let $\ba\in\relo$ and let 
$\neg\ba$ denote its complement. Then for each extension $(\ba,\bc)\in\rel$ 
of $\ba$ the tuple $(\neg\ba,\neg\bc)$ belongs to $\rel$, as $\rel$ is self-complement, 
and $(\neg\ba,\neg\bc)$ is an extension of $\neg\ba$. Therefore $\neg\ba$ has the same number 
of extensions as $\ba$, and so $\neg\ba\in\relo$. Thus, $\relo$ is 
self-complement.
\end{proof}

\begin{lemma}\label{lem:compl-extraction}
Let $\rel$ be a self complement relation that does not belong to $IL_2$ 
(that is, non-affine), then $\Compl_{3,0}\in\mang\rel$ or $\Compl_{1,2}\in\mang\rel$.
\end{lemma}

\begin{proof}
Let $\rel(\vc xn)$ satisfy the conditions of the lemma. There are two cases.

\smallskip

{\sc Case 1.}
$\rel$ does not contain the all-zero tuple.

\smallskip

Observe first that in this case $\mang \rel$ contains the disequality relation. 
Indeed, let $\ba\in\rel$ and let $I\sse[n]$ be the set of indices such that  
$\ba[i]=0$ if and only if $i\in I$. Since the all-zero tuple does not belong 
to $\rel$, $I\ne[n]$. Without loss of generality let $I=[m]$. 
Then it is easy to see that 
$$
\rel(\underbrace{x\zd x}_{\text{$m$ times}},y\zd y)
$$
is the disequality relation.

As $\rel\not\in IL_2$, by Lemma~4.10 of \cite{Creignou01:complexity} 
there are tuples $\ba,\bb,\bc\in\rel$ such that 
$\bd=\ba\oplus\bb\oplus\bc\not\in\rel$. Rearranging the variables these 
tuples can be represented as shown in the table below.
$$
\begin{array}{r|cccccccc|l}
\ba&0\ldots0&0\ldots0&0\ldots0&0\ldots0&1\ldots1&1\ldots1&1\ldots1&1\ldots1&\in\rel\\
\bb&0\ldots0&0\ldots0&1\ldots1&1\ldots1&0\ldots0&0\ldots0&1\ldots1&1\ldots1&\in\rel\\
\bc&0\ldots0&1\ldots1&0\ldots0&1\ldots1&0\ldots0&1\ldots1&0\ldots0&1\ldots1&\in\rel\\
\hline
\bd&0\ldots0&1\ldots1&1\ldots1&0\ldots0&1\ldots1&0\ldots0&0\ldots0&1\ldots1&\not\in\rel\\
\hline
&x\ldots x&y\ldots y&z\ldots z&s\ldots s&t\ldots t&u\ldots u&v\ldots v&w\ldots w
\end{array}
$$
Denote by $\rel'$ the relation obtained from $\rel$ by identifying variables 
as shown in the last row of the table, and then set
\begin{eqnarray*}
\relo(x,y,z,t) &=& \mex s\mex u\mex v\mex w(\rel'(x,y,z,s,t,u,v,w)\\
&& \meet\NEQ(x,w)\meet\NEQ(y,v)\meet\NEQ(z,u)\meet\NEQ(t,s)).
\end{eqnarray*}
Relation $\rel''$ contains tuples $(0,0,0,1),(0,0,1,0),(0,1,0,0)$ but does 
not contain $(0,1,1,1)$, and so does not belong to $IL_2$.

There are 16 cases depending on whether or not tuples (a) $(0,0,1,1)$, 
(b) $(0,1,0,1)$, (c) $(0,1,1,0)$, and (d) $(0,0,0,0)$ belong to $\rel''$ (remember, this 
relation is self complement). If none of them belong to $\rel''$ then 
$\Compl_{3,0}(x,y,z)=\mex t \rel''(t,x,y,z)$. Suppose first 
$(0,0,0,0)\not\in\rel''$. If (a) belongs to $\rel''$ 
then $\Compl_{3,0}(x,y,z)=\rel''(x,x,y,z)$; if (b) is in $\rel''$ then 
$\Compl_{3,0}(x,y,z)=\rel''(x,y,x,z)$; finally, if (c) is in $\rel''$ then 
$\Compl_{3,0}(x,y,z)=\rel''(x,y,z,x)$. Suppose now (d) belongs to $\rel$.
If (a) is not there then $\Compl_{1,2}(x,y,z)=\rel''(x,x,y,z)$. If (a) is also in $\rel$,
then $\Compl_{1,2}(x,y,z)=\rel''(x,y,z,z)$.

\smallskip

{\sc Case 2.}
The all-zero tuple belongs to $\rel$.

\smallskip

Again by Lemma~4.10 of \cite{Creignou01:complexity} there are 
tuples $\ba,\bb,\bc\in\rel$ such that $\bd=\ba\oplus\bb\oplus\bc\not\in\rel$, 
but $\ba$ can be chosen to be the all-zero tuple. Then after rearranging 
variables these tuples can be represented as follows
$$
\begin{array}{r|cccc|l}
\ba&0\ldots0&0\ldots0&0\ldots0&0\ldots0&\in\rel\\
\bb&0\ldots0&0\ldots0&1\ldots1&1\ldots1&\in\rel\\
\bc&0\ldots0&1\ldots1&0\ldots0&1\ldots1&\in\rel\\
\hline
\bd&0\ldots0&1\ldots1&1\ldots1&0\ldots0&\not\in\rel\\
\hline
&x\ldots x&y\ldots y&z\ldots z&t\ldots t
\end{array}
$$
Denote by $\rel'$ the relation obtained from $\rel$ by identifying variables 
as shown in the last row of the table. Relation $\rel'$ contains tuples 
$(0,0,0,0),(0,0,1,1),(0,1,0,1)$ but does not contain $(0,1,1,0)$, and so 
does not belong to $IL_2$.

There are 16 cases depending on whether or not tuples (a) $(0,0,0,1)$, 
(b) $(0,0,1,0)$, (c) $(0,1,0,0)$, and (d) $(1,0,0,0)$ belong to $\rel'$. 
If none of the tuples belong to $\rel'$ or all of them belong to $\rel'$, 
then $\Compl_{2,1}(x,y,z)=\mex t\rel'(t,x,y,z)$. In the first case it is 
1-quantification, and in the second case it is 2-quantification. If exactly one 
of (a) and (b) belongs to $\rel'$ then up to permutation of variables 
$\Compl_{1,2}(x,y,z)=\rel'(x,x,y,z)$.  If exactly one of (a) and (d) belongs 
to $\rel'$ then up to permutation of variables $\Compl_{1,2}(x,y,z)=\rel'(x,y,y,z)$. 
Finally, if exactly one of (c) and (d) belongs to $\rel'$ then up to 
permutation of variables $\Compl_{1,2}(x,y,z)=\rel'(x,y,z,z)$.
\end{proof}

\begin{lemma}\label{lem:self-dual-generation}
If $k+\ell\ge3$ then $\mang{\Compl_{k,\ell}}=IN_2$. 
\end{lemma}

\begin{proof}
Observe first that
\begin{eqnarray}
\Compl_{k,\ell}(\vc x{k+\ell}) &=& \mex y \Compl_{k,\ell+1}(\vc x{k+\ell},y),\nonumber \\
\Compl_{k,\ell}(\vc x{k+\ell}) &=& \mex y (\Compl_{k+1,\ell-1}
(\vc xk,y,x_{k+2},x_{k+\ell})  \label{equ:switching}\\
&& \meet\NEQ(y,x_{k+1})),\quad \text{and} \nonumber\\
\Compl_{k,0}(\vc xk) &=& \mex y \Compl_{k+1,0}(\vc xk,y).\nonumber
\end{eqnarray}
Also,
\begin{eqnarray*}
\lefteqn{\Compl_{k,\ell}(\vc x{k+\ell})}\\
&=&\mex \vc yk \Compl_{k+\ell,0}(\vc yk,x_{k+1}\zd x_{k+\ell+1})
 \meet\NEQ(y_1,x_1)\meet\ldots\meet\NEQ(y_k,x_k)).
\end{eqnarray*}
Since $\NEQ=\Compl_{2,0}$, the equalities above imply that if 
$k'+\ell'\le k+\ell$ then $\Compl_{k',\ell'}\in\mang{\Compl_{k,\ell}}$.

Now it suffices to show that $\Compl_{2k,0}\in\mang{\Compl_{k+1,0}}$. 
We start with the relation given by the following formula
\begin{eqnarray*}
\Phi(\vc x{2k},\vc y{{k\choose 2k}}) &=& \bigwedge_{I=\{\vc ik\}\sse[2k]} 
\Compl_{k+1,0}(x_{i_1}\zd x_{i_k},y_{j_I})\\
&& \meet \bigwedge_{I\sse[2k],|I|=k}\NEQ(y_{j_I},y_{j_{\ov I}}).
\end{eqnarray*}
Here $j_I$ is some enumeration of the $k$-element subsets of $[2k]$. We are 
interested in assignments of $\vc x{2k}$ and the number of ways such an 
assignment can be extended to a satisfying assignment of  $\Phi$. First, observe that 
the only assignments of $\vc x{2k}$ that can not be extended are the all-zero 
and all-one assignment. Second, since $\Phi$ is symmetric with respect of 
permutations of $\{\vc x{2k}\}$ in the sense that for any permutation of this 
set there is a permutation of the $y_i$'s that keeps the formula unchanged, 
the number of extensions of an assignment of $\vc x{2k}$ depends only on 
the number of 0's in the assignment. We will denote this number by 
$N_\Phi(m)$, where $m$ is the number of zeros. Notice that $\Phi$ defines a 
self-complement relation, therefore, we always assume that the number of zeros 
is at least $k$. As is easily seen, if a tuple $\ba$ has $m\ge k$ zeros, 
it can be extended in $N_\Phi(m)=2^{\frac12 {k\choose 2k}-{k\choose m}}$ ways. 
Indeed, $y_I$ is uniquely defined by $\ba$ if $I$ or $\ov I$ is a subset of the 
set of zeros of $\ba$. Otherwise it can take any value independently of the values of 
other variables, except that $y_{j_I}\ne y_{j_{\ov I}}$.

Let $\relo(\vc xk,y)$ be the relation given by: if $x_1=\ldots=x_k$ then $y$ 
can be any, otherwise $y=x_1$. Relation $\relo$ is an intersection of some 
relations $\Compl_{k',\ell'}$ with $k'+\ell'=k+1$. Therefore by (\ref{equ:switching}) 
it belongs to $\mang{\Compl_{k+1,0}}$. Set
$$
\Phi'(\vc x{2k},\vc y{{k\choose 2k}})=\bigwedge_{I=\{\vc ik\}\sse[2k]} 
\relo(x_{i_1}\zd x_{i_k},y_{j_I}),
$$
and consider $\Psi=\Phi\meet\Phi'$, where $\Phi,\Phi'$ have the same variables 
$x_i$, but the sets of the auxiliary variables $y_i$ are disjoint. Observe that 
$N_\Psi(m)=N_\Phi(m)\cdot N_{\Phi'}(m)$. Similarly to $\Phi$, 
$N_{\Phi'}(m)=2^{{k\choose m}}$, provided $m\ge k$. Indeed, variable 
$y_{j_I}$ can be assigned any value if $x_i=0$ for all $i\in I$; otherwise 
$y_{j_I}$ can take only one value. Therefore for any $m\ne0$ 
$$
N_\Psi(m)=2^{\frac12 {k\choose 2k}-{k\choose m}}\cdot 2^{{k\choose m}}
=2^{\frac12 {k\choose 2k}}
$$
and $N_\Psi(0)=0$. Thus 
$\Compl_{2k,0}=\mex (\vc y{k\choose 2k})\Psi$.

It now remains to apply Proposition~3 of \cite{Creignou08:plain} that claims, 
in particular, that the relation $\Compl_{k,\ell}$ constitute a plain basis of $IN_2$.
\end{proof}

\section{Conclusion}

The results of the previous section can be used to reprove some complexity results, 
namely, that of \cite{Dyer10:trichotomy}. If for counting problems $A$ and $B$ 
there are approximation preserving reductions from $A$ to $B$, and from $B$ to $A$, 
we denote it by $A=_{AP} B$. The problem $\NCSP(\IMP)$ plays a special role in 
this result. This problem can also be interpreted as the problem of counting the 
number of independent sets in a bipartite graph, $\#BIS$, or as the problem of counting 
antichains in a partially ordered set \cite{Dyer03:relative}. The problem of counting the number of 
satisfying assignments to a CNF, $\#SAT$, is predictably the most difficult problem 
among counting CSPs.

\begin{theorem}\label{the:boolean trichotomy}
Let $\Gm$ be a set of relations over $\{0, 1\}$. If every relation
in $\Gm$ is affine then $\#CSP(\Gm)$ is in solvable in polynomial time. 
Otherwise if every relation in $\Gm$ is
in $IM_2$ then $\#CSP(\Gm) =_{AP} \#BIS$. Otherwise 
$\#CSP(\Gm) =_{AP} \#SAT$.
\end{theorem}

\begin{proof}
The $\NCSP$ over affine relations can be solved exactly in polynomial time, as it
is proved in \cite{Creignou01:complexity}. If $\Gm$ contains $\OR$ or $\NAND$, 
the problem $\NCSP(\Gm)$ is interreducible with $\#SAT$ by Theorem~3 of
\cite{Dyer03:relative} (observe that the problem \#IS of counting the 
number of independent sets in a graph can be represented as $\NCSP(\NAND)$). 
By Theorems~\ref{the:max-implementation} and~\ref{the:max-post} this leaves
only two max-co-clones to consider, $IM_2$ and $IN_2$. Since $IM_2$ is
generated by $\IMP$ and by Lemma~\ref{lem:monotone-generation}, 
for any $\Gm\sse IM_2$ the problem $\NCSP(\Gm)$ is either polynomial time 
solvable, or is interreducible with $\#BIS$. The remaining max-co-clone, $IN_2$
is generated by $\Compl_{3,0}$ that contains all tuples such that not all
their entries are equal; this is why it is sometimes called the Not-All-Equal 
relation, or NAE. Therefore for any $\Gm\sse IN_2$ such that 
$\Gm\not\sse IL_3$ the problem $\NCSP(\Gm)$ is interreducible with 
$\NCSP(\NAE)$. By \cite{Schaefer78:dichotomy} the decision problem
$\CSP(\NAE)$ is NP-complete. Therefore by Theorem~1 of \cite{Dyer03:relative}
$\NCSP(\NAE)$ is interreducible with $\#SAT$.
\end{proof}

Observe also that some co-clones are not max-co-clones, even those co-clones 
are generated 
(or `determined') by surjective functions. For instance, $IS_{00}$ or $IS_{01}$.
Since on a 2-element set every quantification with $\mexe$ is equivalent to
either existential, or universal quantification, and therefore $\mange\Gm$
can be any set of relations of the form $\Inv(C)$ for a set of surjective 
functions $C$, we obtain the following

\begin{corollary}\label{cor:many-not-one}
There is a set $\Gm$ of relations on $\{0,1\}$ such that $\mang\Gm\ne\mange\Gm$.
\end{corollary}

\bibliographystyle{plain}

\end{document}